\documentclass[11pt, english, american]{article}

\usepackage{amsfonts,amsmath,amssymb,mathrsfs}
\usepackage{dsfont}
\usepackage{graphicx}% Include figure files
\usepackage{color}

\usepackage[bookmarksopen=true,colorlinks,linkcolor = black]{hyperref}

\newcommand{\landau}{\mbox{\begin{scriptsize}$\mathcal{O}$\end{scriptsize}}}

\newcommand{\laa}{\langle\hspace{-0.08cm}\langle}
\newcommand{\raa}{\rangle\hspace{-0.08cm}\rangle}

\newcommand{\LZ}{L^2(\mathbb{R}^2,\mathbb{C})}

\newcommand{\LZN}{L^2(\mathbb{R}^{2N},\mathbb{C})}

\newcommand{\im}{\text{i}}

\newtheorem{theorem}{Theorem}[section]
\newtheorem{lemma}[theorem]{Lemma}
\newtheorem{notation}[theorem]{Notation}

\newtheorem{remark}[theorem]  {Remark}
\newtheorem{definition}[theorem] {Definition}

\newenvironment{proof}{\emph{Proof:}}{\begin{flushright} $ \Box $ \end{flushright}}
\newenvironment{refproof}{\emph{Proof of Theorem }}

\renewcommand{\phi}{\varphi}

\newcommand{\potdiff}{Z}

\begin{document}

\title{Derivation of the
Time Dependent Two Dimensional
 Focusing NLS Equation 
}

\author{Maximilian Jeblick\footnote{jeblick@math.lmu.de Mathematisches Institut, LMU Munich, Germany} \; and Peter Pickl\footnote{pickl@math.lmu.de
Mathematisches Institut, LMU Munich, Germany}}

\maketitle

\begin{abstract}
We present a microscopic derivation of the two-dimensional focusing cubic nonlinear Schr\"odinger equation starting from an interacting $N$-particle system of Bosons. 
The interaction potential we consider is given by $W_\beta(x)=N^{-1+2 \beta}W(N^\beta x)$ for some
spherically symmetric and compactly supported potential $W \in L^\infty (\mathbb{R}^2, \mathbb{R})$. The class of initial wave functions is chosen such that the variance in energy is small.
Furthermore, we assume that the Hamiltonian $ H_{W_\beta, t}=-\sum_{j=1}^N \Delta_j+\sum_{1\leq  j< k\leq  N} W_\beta(x_j-x_k) +\sum _{j=1}^N A_t(x_j)$ fulfills stability of second kind, that is $ H_{W_\beta, t} \geq -CN$.
We then prove the convergence of the reduced density matrix corresponding to the exact time evolution to the projector onto the solution of the corresponding nonlinear Schr\"odinger equation in either Sobolev trace norm, if $\|A_t\|_p < \infty$ for some $p>2$, 
 or in trace norm, for more general external potentials. 
For trapping potentials of the form $A(x)=C |x|^s\; , C>0$, the condition
$ H_{W_\beta, t} \geq -CN$ can be fulfilled for a certain class of interactions $W_\beta$, for all $0 < \beta < \frac{s+1}{s+2}$, see
\cite{lewin}.

\end{abstract}
\newpage
\tableofcontents\newpage

%-------------------------

\newpage
\section{Introduction}
During the last decades the experimental realization and the theoretical investigation of Bose-Einstein condensation (BEC) regained a considerable amount of attention. Mathematically, there is a steady effort to describe both the dynamical as well as the statical properties of such condensates. While the principal mechanism of BEC is similar for many different systems, the specific effective description of such a system depends strongly on the model one studies.
In this paper we will focus on a dilute, two dimensional system of bosons with attractive interaction.

Let us first define the $N$-body quantum problem we have in mind. 
The evolution of $N$ interacting bosons is described by a time-dependent wave-function $\Psi_t \in L^2_{s}(\mathbb{R}^{2N}, \mathbb{C}),
\| \Psi_t \|=1$ (Here and below norms without index $\|\cdot\|$ always denote the $L^2$-norm on the appropriate Hilbert space.).
$L^2_{s}(\mathbb{R}^{2N}, \mathbb{C})$ denotes the set of all $\Psi \in  L^2(\mathbb{R}^{2N}, \mathbb{C})$
 which are symmetric under pairwise permutations of the variables 
$x_1, \dots, x_N \in \mathbb{R}^2$. 
$\Psi_t$ solves the  $N$-particle 
Schr\"odinger equation 
\begin{align}
\label{schroe}
 \im\partial_t \Psi_t =  H_{W_\beta,t} \Psi_t 
 \;,
\end{align}
where the time-dependent Hamiltonian $ H_{W_\beta,t}$ is given by
\begin{align}
\label{hamiltonian}
 H_{W_\beta, t}=-\sum_{j=1}^N \Delta_j+\sum_{1\leq  j< k\leq  N} W_\beta(x_j-x_k) +\sum _{j=1}^N A_t(x_j)
 \; .
\end{align} 
The scaled potential $W_\beta(x)= N^{-1+2 \beta} W(N^\beta x)$,
$W \in L^\infty_c(\mathbb{R}^2, \mathbb{R})$ describes a strong, but short range
potential acting on the length scale of order $N^{-\beta}$ (we assume $W$ to be compactly supported). 
The external potential $A_t \in L^p(\mathbb{R}^2, \mathbb{R})$ for some $p>1$ is used as an
external trapping potential. Below, we will comment on different choices for $A_t$ in more detail. 
In general, even for small particle numbers $N$, the evolution equation \eqref{schroe} cannot be solved directly nor numerically for $\Psi_t$. Nevertheless, for a certain class of initial conditions $\Psi_0$ and certain interactions $W$, which we will make precise in a moment, it is possible to derive an approximate solution of \eqref{schroe} in the trace class topology of reduced density matrices.

Define the one particle reduced density matrix $\gamma^{(1)}_{\Psi_0}$ of $\Psi_0$ with integral kernel
$$\gamma^{(1)}_{\Psi_0}(x,x')=\int_{\mathbb{R}^{2N-2}} \Psi_{0}^*(x,x_2,\ldots,x_N)\Psi_{0}(x',x_2,\ldots,x_N)d^2x_2\ldots d^2x_N \;. $$
To account for the physical situation of a Bose-Einstein condensate, we assume complete condensation in the limit of large particle number $N$. This amounts to
assume that, for $N \rightarrow \infty$,
 $\gamma^{(1)}_{\Psi_0}  \rightarrow |\phi_0\rangle\langle\phi_0|$ in trace norm for some $\phi_0
\in L^2(\mathbb{R}^2,\mathbb{C}) 
  , \|\phi_0\|=1$.
   Define $a=\int_{\mathbb{R}^2} d^2x W(x)$ (throughout this paper, $a$ will always denote the integral over $W$).
Let $\phi_t$ solve  the nonlinear Schr\"odinger equation
\begin{align}
\label{nls}
\im \partial_t
\phi_t=\left(-\Delta +A_t\right) \phi_t+
a
|\phi_t|^2\phi_t=:h_{ 
a,t
}^{\text{NLS}}\phi_t
\end{align}
with initial datum $\phi_0$. 
 Our main goal is to show the persistence of condensation over time. 
In particular, we prove that the time evolved reduced density matrix $\gamma^{(1)}_{\Psi_t}$
 converges to $ |\phi_t\rangle\langle\phi_t|$ in trace norm as $N \rightarrow \infty$ with convergence rate of order $N^{-\eta}$ for some explicitly computable $\eta>0$, see Lemma \ref{gammalemma}. 
Assuming $W \in L^\infty_c(\mathbb{R}^2, \mathbb{R})$, such that $W$ is spherically symmetric and  $W \geq 0$, we were possible to show convergence $\gamma^{(1)}_{\Psi_t} \rightarrow| \phi_t \rangle \langle \phi_t |$ in trace norm for all $\beta>0$, under suitable conditions on
$\Psi_0$ and $\phi_0$ \cite{jeblick}, see also \cite{schlein2d} for a prior result.
The problem becomes more delicate for interactions which are not nonnegative, especially if \eqref{nls} is focusing, which means $a<0$. 
For strong, attractive potentials, it is known that the condensate collapses in the limit of large particle number.
To prevent this behavior, our proof needs stability of second kind for the Hamiltonian $ H_{W_\beta, t}$, that is, we assume $ H_{W_\beta, t} \geq -CN$. If $W_\beta$ is partly or purely nonpositive, this assumption gets highly nontrivial for higher $\beta$. 
For $\beta \leq 1/2$, the inequality
\begin{align}
& \inf_{\Psi \in L^2(\mathbb{R}^{2N}, \mathbb{C}), \| \Psi \|=1}
\frac{\laa \Psi,  H_{W_\beta, t} \Psi \raa}{N }
\nonumber 
\\
\label{nlsenergy}
\geq &
\inf_{\phi \in L^2(\mathbb{R}^{2}, \mathbb{C}), \| \phi \|=1}
\left(
\int_{\mathbb{R}^2} d^2x 
\left(
	|\nabla \phi(x)|^2+ A_t(x) |\phi(x)|^2
	+
	\frac{1}{2}
	\int_{\mathbb{R}^2} d^2x 
	 |\phi(x)|^2 (NW_\beta \ast | \phi|^2)(x)
\right)
\right)
\\
\nonumber
-&
\landau (1)
-C N^{2 \beta-1},
\end{align}
which was proven in \cite{lewinMeanField}, shows $ H_{W_\beta,t} \geq -CN$, if \eqref{nlsenergy}, which is the ground state energy of the nonlinear Hartree functional, is bounded from below uniformly in $N$
(Here and in the following,  $\laa\cdot,\cdot\raa$ denotes the scalar product on $\LZN$.).
Assuming $A_t \geq -C$,
this is the case if
\begin{align}
\label{condition}
\inf_{\phi \in H^1(\mathbb{R}^2, \mathbb{C})}
\left(
\frac{\int_{\mathbb{R}^2} d^2x |\phi(x)|^2 ( |\phi|^2 \ast W)(x)}
{ \|\phi\|^2  \|\nabla \phi \|^2}
\right)
>-1
\end{align}
holds  \cite{lewin}. 
Assuming condition \eqref{condition} together with $
A_t \in L^1_{\text{loc}} (\mathbb{R}^2, \mathbb{R}),\;
A_t(x) \geq C |x|^s$, $s >0$, 
 stability of second kind was subsequently proven for all $0<\beta< \frac{s+1}{s+2}$ \cite{lewin}.
 In particular, it was shown that the ground state energy per particle of $H_{W_\beta, t}$ is then given (in the limit $N \rightarrow \infty$) by the corresponding nonlinear Schr\"odinger functional; see also \cite{lewinneu} for an earlier result which also treats the one- and three-dimensional cases.
 \\
 Condition \eqref{condition} thus restricts the set of interactions $W$. Indeed, stability of the second kind fails if
 \begin{align}
\inf_{\phi \in H^1(\mathbb{R}^2, \mathbb{C})}
\left(
\frac{\int_{\mathbb{R}^2} d^2x |\phi(x)|^2 ( |\phi|^2 \ast W)(x)}
{ \|\phi\|^2  \|\nabla \phi \|^2}
\right)
<-1 ,
\end{align}
see \cite{lewinneu, lewin} for a nice discussion.
Let $W^-$ denote the negative part of $W$ and let $a^*>0$ denote
the optimal constant of the Gagliardo-Nirenberg inequality
\begin{align*}
\left(\int_{\mathbb{R}^2} d^2x | \nabla u(x) |^2 \right)
\left(\int_{\mathbb{R}^2} d^2y |  u(y) |^2 \right)
\geq
\frac{a^*}{2}
\left(\int_{\mathbb{R}^2} d^2x | u(x) |^4 \right) \;.
\end{align*}
It is then easy to prove that $  
\int_{\mathbb{R}^2} d^2x |W^-(x)|
 < a^*$ implies condition \eqref{condition}.
On the other hand, \eqref{condition} implies $a > - a^*$.
While \eqref{condition} is in general a weaker condition than  $  
\int_{\mathbb{R}^2} d^2x |W^-(x)|
 < a^*$, 
for $W \leq 0$, they are equivalent.
Consequently, for nonpositive $W$ and for $a<-a^*$, the nonlinear Hartree functional is not bounded from below in the limit $N \rightarrow \infty$, which in particular implies that $ H_{W_\beta, t}$ is not stable of the second kind.
It is also known that $a^*$ is the critical threshold for blow-up solutions, that is, 
for $a \leq-a^*$ the solution of \eqref{nls} may blow up in finite time  \cite{ carles2, carles, cazenave, killip, tao, weinstein}.
\\
The condition $ H_{W_\beta, t} \geq -CN$ is needed in our proof to control the kinetic energy of those particles which leave the condensate, see Lemma \ref{varianzlemma}. 
In prior works, it was necessary to control 
the quantity $\|\nabla_1 q_1 \Psi_t\|$ \footnote{
Here, $q_1$ denotes the complement of the projection onto $\phi_t(x_1)$, see  Definition \ref{defpro} below.
}  sufficiently well
in order to show convergence of the reduced density matrices for large $\beta$,
using the method introduced in \cite{pickl1}.
While it is possible to obtain an a priori estimate of $\|\nabla_1 q_1 \Psi_t\|$ for repulsive interactions, 
it is not obvious how one could generalize this estimate for nonpositive $W$.
 Our strategy to overcome this difficulty is to control the quantity
 $
 \| q_2\nabla_1 \Psi_t \|
 $ instead. Under some  natural assumptions (see (A2), (A4) and (A5) below),  it is possible to obtain a
sufficient 
  bound of $ \| q_2\nabla_1 \Psi_t \|
  $, if initially the variance of the energy
 \begin{align}
 \text{Var}_{ H_{W_\beta,0}}(\Psi_0)
=
\frac{1}{N^2}
\laa \Psi_0, 
\left(
 H_{W_\beta, 0}
-
\laa \Psi_0, 
 H_{W_\beta, 0}
\Psi_0 \raa
\right)^2
\Psi_0 \raa
 \end{align}
 is small and $H_{W_\beta, t}$ is stable of second kind. 
For product states $\Psi_0=\phi^{\otimes k}$, with $\phi$ regular enough, a straightforward calculation yields
$
 \text{Var}_{ H_{W_\beta, 0}}(\Psi_0)\leq C(1+N^{-1+\beta}+N^{-2+2\beta})$, see Lemma \ref{varianzlemma}. Therefore, at least for $\beta<1$, there exist initial states $\Psi_0$, for which the variance is small.
 The strategy to control $ \| q_2\nabla_1 \Psi_t \|$ instead of 
 $\|\nabla_1 q_1 \Psi_t\|$ by means of the energy variance was already used in \cite{leopold} where the derivation of the Maxwell-Schr\"odinger equations from the Pauli-Fierz Hamiltonian was shown\footnote{
 We would like to thank Nikolai Leopold 
 for pointing out to us 
the idea to use the variance of the energy $\text{Var}_{ H_{W_\beta, t}} (\Psi_t)$ 
in our estimates and how to include time-dependent external potentials.}.  
 Adopting this method, we are able to prove convergence of $\gamma^{(1)}_{\Psi_t}$
  to $ |\phi_t\rangle\langle\phi_t|$ in trace norm as $N \rightarrow \infty$ with convergence rate of order $N^{-\eta} \;, \eta>0$, if the assumptions (A1)-(A5) (see below) are fulfilled. We like to remark that 
it is unclear if the assumptions (A1) (stability of second kind of $H_{W_\beta, t}$) and (A3) (smallness of  $\text{Var}_{ H_{W_\beta, 0}}(\Psi_0)$)  can be fulfilled for $\beta \geq 1$.

A stronger statement than convergence in trace norm is convergence in Sobolev trace norm.
For external potentials $A_t \in L^p( \mathbb{R}^2, \mathbb{R})$, with $p \in ]2, \infty]$, we are able to 
show 
\begin{align}
\label{energynorm_convergence}
\lim_{N \rightarrow \infty}
\text{Tr}
\left|
\sqrt{1-\Delta}
( \gamma_{\Psi_t}^{(1)}- | \phi_t \rangle \langle \phi_t |)
\sqrt{1-\Delta}
\right|=0 ,
\end{align}
 if initially the energy per particle
$ N^{-1} \laa \Psi_0, H_{W_\beta, 0} \Psi_0 \raa$ is close to the NLS energy
$
\langle \phi_0, 
\left(
-\Delta 
+
\frac{a}{2}
 |\phi_0|^2
+
 A_0 
 \right)
 \phi_0 \rangle$. 
To obtain this type of convergence, we adapt some recent results of \cite{anapolitanos, picklnorm},  where a similar statement was proven.

The rigorous derivation of effective evolution equations has a long history, see e.g. \cite{ benedikter, brennecke, SchleinNorm, erdos1, erdos2, erdos3, erdos4, knowles,  marcin1, marcin2, picklgp3d, pickl1, rodnianskischlein}
 and references therein.
 In particular, for the two-dimensional case we consider, it has been proven, for $0 <\beta <3/4$ and $W$ nonnegative, that $\gamma^{(1)}_{\Psi_t}$
 converges to $ |\phi_t\rangle\langle\phi_t|$ as $N \rightarrow \infty$ \cite{schlein2d}.
 This was later extended by us to hold for all $\beta>0$ \cite{jeblick}.
For $A(x)=|x|^ 2$ and $W \leq 0$ sufficiently small such that $ H_{W_\beta, t} \geq -CN$, the respective convergence has been proven in \cite{chen2d} for $0 < \beta < 1/6$. 
 One key estimate of the proof was to show the stability condition $ H_{W_\beta, t} \geq -CN$. The authors note that their proof actually works for all $0 <\beta < 3/4$, if 
 $ H_{W_\beta, t} \geq -CN$ holds.
As mentioned, this was subsequently proven by  \cite{lewin} in a more general setting.
Recently, 
 the validity of the Bogoloubov approximation for 
the two-dimensional attractive bose gas was 
 shown in \cite{marcin3} for $0 < \beta <1$.
 In contrast to our result, the authors were actually able to achieve norm convergence and did not
  need to impose the stability condition
  $ H_{W_\beta, t} \geq -CN$, but only required the bound
  $  
\int_{\mathbb{R}^2} d^2x |W^-(x)|
 < a^*$. 
They then use some refined localization method on the number of particles. 
 This strategy enables them to analyze the dynamics without any external field.
 It would be nice to achieve a better understanding on the relationship between stability of second kind and the bound on the negative part of $W$, as stated above. We plan to come back to this point in the future.
 We  want to emphasize that norm convergence is a stronger statement than convergence in the topology of reduced densities. 
However, convergence in Sobolev trace norm as defined in  \eqref{energynorm_convergence}
does in general not follow from norm convergence.
\\
It is also possible to consider the two dimensional Bose gas in the regime where 
the short scale correlation structure is important for the dynamics. The scaling to consider in such a case is given by $
e^{2N} W(e^N x)
$. We refer the reader to \cite{jeblick, lssy}.

For  $0<\beta<1/4$, it can be verified that the methods presented in \cite{pickl2}, where the attractive three dimensional case is treated, can be applied, assuming some regularity conditions on $\phi_t$ (the corresponding conditions for the three dimensional system were proven in 
\cite{chong}). Interestingly, this proof does 
not restrict the strength of the nonpositive potential nor does it require stability of second kind, but rather assumes a sufficiently fast convergence of $\gamma^{(1)}_{\Psi_0}$ to $ |\phi_0\rangle\langle\phi_0|$. Therefore, one can prove BEC in two dimensions for $\beta<1/4$ and arbitrary strong attractive interactions for times for which the solution $\phi_t$ exists and is regular enough, that is, before some possible blow-up.

%-------------------------

\section{Main result}

We will bound  expressions which are uniformly bounded in $N$ and $t$ by some constant $C>0$.
We will not distinguish constants
 appearing in a sequence of estimates, i.e. in $X\leq  CY\leq  CZ$ the constants may differ.
 We denote by  $\laa\cdot,\cdot\raa$  the scalar product on $\LZN$
and by $\langle\cdot,\cdot\rangle$  the scalar product on $\LZ$.

We will require the following assumptions:
\begin{itemize}
\item[(A1)]
For $\beta>0$, let $W_\beta$ be given by $W_\beta(x)= N^{-1+2 \beta} W(N^{\beta}x)$, for $W \in 
L^\infty_c (\mathbb{R}^2,\mathbb{R})$, $W$ spherically symmetric.
We assume that there exist constants $0<\epsilon, \mu<1$ such that
\begin{align*}
 H_{W_\beta, t}^{(\epsilon,\mu)}=
(1-\epsilon)\sum_{k=1}^N (-\Delta_k)+ \sum_{i<j} W_\beta(x_i-x_j) +(1-\mu) \sum_{k=1}^N A_t(x_k) \geq -CN
\;.
\end{align*}
\item[(A2)]
For any real-valued function $f$, decompose $f(x)=f^+(x)-f^-(x)$ with $f^+(x), f^-(x) \geq 0$, such 
that the supports of $f^+$ and $f^-$ are disjoint. 
We assume that $A_t^- \in L^\infty(\mathbb{R}^2, \mathbb{R})$. Furthermore, we assume that
$A_t \in H^ 2(\mathbb{R}^ 2, \mathbb{R})$ is differentiable with respect to $t$ and fulfills
\begin{align*}
\dot{A}_t \in L^\infty (\mathbb{R}^2, \mathbb{R}) 
\;,
\nabla \dot{A}_t \in L^\infty (\mathbb{R}^2, \mathbb{R}) \;,
\Delta \dot{A}_t \in L^\infty (\mathbb{R}^2, \mathbb{R})
\end{align*} 
for all $t \in \mathbb{R}$.
\item[(A3)]
For any $s \in \mathbb{R}$,
we denote for $k \in \mathbb{N}$ the domain of the self-adjoint operator $ (H_{W_\beta, s})^k$ by $\mathcal{D}( (H_{W_\beta, s})^k)$.
Define the energy variance $\text{Var}_{ H_{W_\beta, s}}:
\mathcal{D}( (H_{W_\beta, s})^2)
\to \mathbb{R}^+$
as
$$
\text{Var}_{ H_{W_\beta, s}}(\Psi)
=
\frac{1}{N^2}
\laa \Psi, 
\left(
 H_{W_\beta, s}
-
\laa \Psi, 
 H_{W_\beta, s}
\Psi \raa
\right)^2
\Psi \raa
\;.
$$
We then require $\text{Var}_{ H_{W_\beta, 0}}(\Psi_0) \leq CN^{-\delta}$ for some $\delta>0$.
\item[(A4)]
Let $\phi_t$ the solution to $ i \partial_t \phi_t= h^{\text{NLS}}_{a, t} \phi_t, \; \| \phi_0 \|=1$. 
We assume that $\phi_t \in H^{4}(\mathbb{R}^2, \mathbb{C})$.
\item[(A5)]
Assume that the energy per particle
$$N^ {-1} |\laa \Psi_0,  H_{W_\beta, 0} \Psi_0 \raa | \leq C $$ and
the NLS energy
$$
\left|
\langle \phi_0, 
\left(
-\Delta 
+
\frac{a}{2}
 |\phi_0|^2
+
 A_0 
 \right)
 \phi_0 \rangle
\right| 
 \leq C
$$
are bounded uniformly in $N$ initially.
\item[(A5)']
Assume that there exists a $\delta>0$, such that
\begin{align*}
 \left|
 N^ {-1}\laa \Psi_0,  H_{W_\beta, 0} \Psi_0 \raa 
 -
 \langle \phi_0, 
\left(
-\Delta 
+
\frac{a}{2}
 |\phi_0|^2
+
 A_0 
 \right)
 \phi_0 \rangle 
 \right|\leq C N^{-\delta}.
\end{align*}

\end{itemize}
\begin{remark}
\begin{enumerate}
\item
Note that (A1) together with (A2) directly implies $ H_{W_\beta, t}^{(\epsilon,0)} \geq -CN
\;,  H_{W_\beta, t}^{(0,\mu)} \geq -CN$ and $ H_{W_\beta, t}= H_{W_\beta, t}^{(0,0)} \geq -CN$.
As mentioned in the introduction, (A1) and (A2) are fulfilled for $A(x) \geq C|x|^s \;, s >0$ 
for any $0<\beta< \frac{s+1}{s+2}$, assuming \eqref{condition}.
\cite{lewin}.
\item
Assuming $\Psi_0= \phi_0^{\otimes N}$ with 
$ 
\phi_0 \in W^{2,\infty}(\mathbb{R}^2,\mathbb{C}) \cap
H^1(\mathbb{R}^2,\mathbb{C}), \| \phi_0 \|=1
$
such that
$
 \langle \phi_0, A_0 \phi_0 \rangle
+
\langle \phi_0, A_0^2 \phi_0 \rangle \leq C
$,
it follows that 
$
\text{Var}_{ H_{W_\beta, 0}}(\Psi_0)
\leq
C(N^{-1+\beta}+ N^{-2+2 \beta})
$, see Lemma \ref{varianceproduct}; and hence (A3) is then valid for all $0<\beta<1$. 
\item
For $A_t \in \lbrace 0, |x|^2 \rbrace$, (A4) follows from the persistence of regularity of solutions, assuming $
\phi_0 \in H^4(\mathbb{R}^2,\mathbb{C}) ,\; \| A_0^2 \phi_0\| <\infty $, see Appendix \ref{solutionTheory}. However, for regular enough external potentials, $a>-a^*$ and regular enough $\phi_0$ we believe (A4) to be valid, too.
\item
It is interesting to note that both (A1) and (A3) can be fulfilled for $0<\beta<1$, while it is unclear if they hold for $\beta\geq 1$.
\end{enumerate}
\end{remark}

We now state our main Theorem:
\begin{theorem}\label{theo}
Let $\Psi_0 \in L^2_{s}(\mathbb{R}^{2N}, \mathbb{C})
\cap \mathcal{D}( (H_{W_\beta, 0})^2)$ with $\|\Psi_0\|=1$. Let $\phi_0  \in L^{2}(\mathbb{R}^2,\mathbb{C})$ with $\|\phi_0\|=1$ and assume 
$\lim\limits_{N\to\infty}
\left(
N^{\delta}
\text{Tr} | \gamma^{(1)}_{\Psi_0}-|\phi_0\rangle\langle\phi_0| |
\right)
=0 $ for some $\delta > 0$.
Let $\Psi_t$ the unique solution to $i \partial_t \Psi_t
=  H_{W_\beta, t} \Psi_t$ with initial datum $\Psi_0$.
Let $\phi_t$ the unique solution to $i \partial_t \phi_t
= h^{\text{NLS}}_{a, t} \phi_t$ with initial datum $\phi_0$.
\begin{enumerate}

\item (Convergence in trace norm)
Assume (A1)-(A5).
Then, for any $t>0$ 
\begin{equation}
\label{convergenls}
\lim_{N\to\infty}\gamma^{(1)}_{\Psi_t}=|\phi_t\rangle\langle\phi_t|
\end{equation} in trace norm.
\item (Convergence in Sobolev trace norm)
Assume (A1)-(A5) and (A5)'. Furthermore, assume that
$A_t \in L^p(\mathbb{R}^ 2, \mathbb{R})$ holds for some $p \in ]2, \infty]$ and for all $t \in \mathbb{R}$.
Then, for any $t>0$ 
\begin{equation}
\lim_{N \rightarrow \infty}
\text{Tr}
\left|
\sqrt{1-\Delta}
( \gamma_{\Psi_t}^{(1)}- | \phi_t \rangle \langle \phi_t |)
\sqrt{1-\Delta}
\right|
=
0.
\end{equation}

\end{enumerate}

\end{theorem}
\begin{remark}
\begin{enumerate}

\item 
It is well known that convergence of $\gamma^{(1)}_{\Psi_t}$
 to $ |\phi_t\rangle\langle\phi_t|$ in trace norm is equivalent to the respective convergence in operator norm since $|\phi_t\rangle\langle\phi_t|$ is a rank-1-projection, see
Remark 1.4. in \cite{rodnianskischlein}. 
Other equivalent definitions of asymptotic 100\% condensation can be found in \cite{michelangeli}.
Furthermore, the convergence of the one-particle reduced density matrix
$\gamma^{(1)}_{\Psi_t} \rightarrow| \phi_t \rangle \langle \phi_t |$ in trace norm
 implies convergence of any $k$-particle reduced density matrix $\gamma^{(k)}_{\Psi_t}$ against $| \phi_t^{\otimes k} \rangle \langle \phi_t^{\otimes k} |$ in trace norm as $N \rightarrow \infty$ and $k$ fixed, see for example \cite{knowles}.

\item
 In our proof we will give explicit error estimates in terms of the particle number $N$. We shall show that the rate of convergence is of order
$N^{-\delta}$ for some $\delta>0$, assuming that initially
$\text{Tr}|\gamma^{(1)}_{\Psi_0}- |\phi_0\rangle\langle\phi_0|| 
\leq
C N^{-2 \delta}
$ holds.
 See \eqref{tracestimate} for the precise error estimate.

\item Under assumption (A2), the  
domains $\mathcal{D}( H_{W_\beta, t})$ and $\mathcal{D}( (H_{W_\beta, t})^2)$ of the time dependent Hamiltonian $ H_{W_\beta, t}$ are time-invariant, see Appendix \ref{selfadjoint}.
Therefore, the condition $\Psi_0 \in \mathcal{D}( (H_{W_\beta, 0})^2)$ is sufficient to define and to differentiate
the variance of the energy $\text{Var}_{ H_{W_\beta, t}} (\Psi_t)$.

\item
For $A_t(x)=|x|^2$, $0< \beta< 3/4$ and under condition \eqref{condition},
 the assumptions (A1)-(A5) can be fulfilled by
choosing $\Psi_0= \phi_0^{\otimes N}$ with $\phi_0$ regular enough.
We are therefore able to reproduce the result presented in \cite{chen2d} under slightly different assumptions, using the result of \cite{lewin} which implies (A1). 
 
\item
For external potentials $A_t$ which are bounded from below, assumption (A1) has been proven for all $0< \beta \leq 1/2$, under the condition \eqref{condition} \cite{lewinMeanField}.
We are therefore able to control the  
 convergence of $\gamma^{(1)}_{\Psi_t}$ to $ |\phi_t\rangle\langle\phi_t|$ in Sobolev trace norm as $N \rightarrow \infty$ for $0 < \beta \leq 1/2$.
 
 \item
In our estimates, we need the regularity conditions 
\begin{align*}
\|  \Delta \phi_t \|_\infty < \infty,
\;
\| \nabla \phi_t \|_\infty < \infty,
\;
\|  \phi_t \|_\infty < \infty,
\;
\| \nabla \phi_t \| < \infty,
\;
\| \Delta\phi_t \| < \infty \;.
\end{align*}
That is, we need $\phi_t \in H^2(\mathbb{R}^2,\mathbb{C}) \cap W^{2,\infty}(\mathbb{R}^2,\mathbb{C})$.
Then, $ \| \Delta|\phi_t| ^2 \|$ which also appears in our estimates, can be bounded by
\begin{align*}
 \Delta|\phi_t|^ 2
 =&
 \phi^*_t \Delta \phi_t
  +
 \phi_t  \Delta\phi^*_t
   +
  2   (\nabla \phi^*_t) \cdot (\nabla \phi_t)
  \\
 \| \Delta|\phi_t| ^2 \|
 \leq &
2
 \| \Delta \phi_t \| \| \phi_t \| _\infty
 +
 2
 \| \nabla \phi_t \|  \| \nabla \phi_t \| _\infty
\end{align*}
Recall the Sobolev embedding Theorem, which implies in particular
$H^k(\mathbb{R}^2,\mathbb{C})= W^{k,2}(\mathbb{R}^ 2,\mathbb{C}) \subset C^{k-2}(\mathbb{R}^2,\mathbb{C})$. 
If $\phi \in C^2(\mathbb{R}^2,\mathbb{C}) \cap H^2(\mathbb{R}^2,\mathbb{C})$, then $\phi \in W^{2,\infty}(\mathbb{R}^2,\mathbb{C})$ follows since both $\phi$ and $\nabla \phi$ have to decay at infinity.
 Thus, $\phi_t \in H^4(\mathbb{R}^2,\mathbb{C})$ implies $\phi_t \in H^2(\mathbb{R}^2,\mathbb{C}) \cap W^{2,\infty}(\mathbb{R}^2, \mathbb{C})$, which suffices for our estimates \footnote{ Actually, 
 $\phi_t \in H^{3+\epsilon}(\mathbb{R}^2,\mathbb{C})$ for some $\epsilon>0$ would suffice for our estimates. Note that it is reasonable to expect persistence of regularity of $\phi_t$ assuming $
\phi_t \in L^\infty(\mathbb{R}^2,\mathbb{C})$, see also Appendix \ref{solutionTheory}. }.

\end{enumerate}
\end{remark}

\section{Proof of Theorem \ref{theo} (a)}\label{secpro}

We fix the notation we are going to employ during the rest of the paper.
\begin{notation}
\begin{enumerate}

\item
We will denote the operator norm defined for any linear operator $f:\LZN\to\LZN$ by
$$\|f\|_{\text{op}}=\sup_{\psi \in \LZN, \|\Psi\|=1}\|f\Psi\|\;.$$

 \item
We will denote for any multiplication operator 
 $ F:L^2(\mathbb{R}^2, \mathbb{C}) \rightarrow L^2(\mathbb{R}^2, \mathbb{C}) $
 the corresponding operator 
 $$
 \mathds{1}^{\otimes (k-1)} \otimes F \otimes  \mathds{1}^{\otimes (N-k)} :
 L^2(\mathbb{R}^{2N}, \mathbb{C}) \rightarrow L^2(\mathbb{R}^{2N}, \mathbb{C}) 
 $$
 acting on the $N$-particle Hilbert space
 by $F(x_k)$. In particular, we will use, for any $ \Psi,\Omega \in  L^2(\mathbb{R}^{2N}, \mathbb{C})$ the notation
 $$
 \laa \Omega,  \mathds{1}^{\otimes (k-1)} \otimes F \otimes  \mathds{1}^{\otimes (N-k)}\Psi \raa
 =
 \laa \Omega, F(x_k) \Psi \raa 
 \;.
 $$
 In analogy, for any two-particle multiplication operator $K:L^2(\mathbb{R}^2, \mathbb{C}) ^{\otimes 2}\rightarrow L^2(\mathbb{R}^2, \mathbb{C})^{\otimes 2} $, we denote the operator acting on any $ \Psi \in  L^2(\mathbb{R}^{2N}, \mathbb{C})$ by multiplication
 in the variable $x_i$ and $x_j$ by $K(x_i,x_j)$. In particular, we denote
 $$
  \laa \Omega, K(x_i,x_j) \Psi \raa 
  =
\int_{\mathbb{R}^{2N}} 
K(x_i, x_j)
 \Omega^*(x_1,\ldots,x_N) \Psi  (x_1,\ldots,x_N) d^2x_1 \dots d^2x_N   \;.
 $$
 \item
 We will denote by $\mathcal{K}(\phi_t)$ a constant depending on time, via
 $
 \|\dot{A}_t\|_\infty ,\; \|A^-_t \|_\infty \;,$ $ \int_0^t ds \| \dot{A}_s \|_\infty
 $ and  $\| \phi_t \|_{H^4}$. As mentioned above, we make use of the embedding
 $ 
  H^2(\mathbb{R}^2,\mathbb{C}) \cap W^{2,\infty}(\mathbb{R}^2, \mathbb{C})
  \subseteq H^4(\mathbb{R}^2,\mathbb{C}) 
 $.

\end{enumerate}
\end{notation}

The method we  use in this paper is introduced in detail in  \cite{pickl1} and was generalized to derive various mean-field equations. 
Our proof is based on \cite{jeblick, picklgp3d}. In \cite{jeblick} we proved the equivalent of Theorem \ref{theo} for nonnegative potentials and for all $\beta>0$. However, since we are not covering the two dimensional Gross-Pitaevskii regime where one considers an exponential scaling of the interaction, our estimates are less involved. Furthermore,
we adopt some ideas which were first presented in \cite{leopold}.
Heuristically speaking, the method we are going to employ is based on the idea of counting for each time $t$ the relative number of those particles which are
not in the state $\phi_t$. 
It is then possible to show that the rate of particles which leave the condensate is small, if initially almost all particles are in the state $\phi_0$.  
In order to compare the exact dynamic, generated by $ H_{W_\beta, t}$, with the effective dynamic, generated by $h^{\text{NLS}}_{a, t}$, we define the projectors $p^\phi_j$ and $q^\phi_j$.
\begin{definition}\label{defpro}
Let $\phi\in L^{2}(\mathbb{R}^2,\mathbb{C})$ with $\| \phi  \|=1$.
For any $1\leq  j\leq  N$ the
projectors $p_j^\phi:\LZN\to\LZN$ and $q_j^\phi:\LZN\to\LZN$ are defined as
\begin{align} p_j^\phi\Psi=\phi(x_j)\int_{\mathbb{R}^2} \phi^*(\tilde{x}_j)\Psi(x_1,\ldots, \tilde{x}_j,\dots,x_N)d^2\tilde{x}_j\;\;\;\forall\;\Psi\in\LZN
\end{align}
and $q_j^\phi=1-p_j^\phi$.
We shall also use, with a slight abuse of notation, the bra-ket notation
$p_j^\phi=|\phi(x_j)\rangle\langle\phi(x_j)|$.

\end{definition}
For ease of notation, we will often omit the upper index $\phi$ on $p_j$,
$q_j$, except where their $\phi$-dependence plays an important role.
Our key strategy is to define a convenient functional $\alpha$, depending on $\Psi_t$ and $\phi_t$, such that $
\lim\limits_{N \rightarrow \infty}  \alpha( \Psi_t, \phi_t)
 = 0$ implies Theorem \ref{theo}.
\begin{definition}\label{defm}
Let $\Psi \in L^2 (\mathbb{R}^{2N}, \mathbb{C})
\cap \mathcal{D}(  (H_{W_\beta, \cdot})^2)
$ and let $\phi \in L^2(\mathbb{R}^2,\mathbb{C}),\; \|\phi \|=1$.
Define
\begin{align}
\nonumber
&\alpha: L^2(\mathbb{R}^{2N}, \mathbb{C}) \times
L^2(\mathbb{R}^{2}, \mathbb{C}) \rightarrow \mathbb{R}^+ ,
\\
&\alpha(\Psi,\phi)=
\laa\Psi,q^\phi_1\Psi\raa\
+
\text{Var}_{ H_{W_\beta, \cdot}}(\Psi)
\;.
\end{align}

\end{definition}
Using a general strategy, we will estimate the time derivative $\frac{d}{dt}
\alpha(\Psi_t, \phi_t)$. In particular, we
show that
\begin{align*}
\frac{d}{dt}
\alpha(\Psi_t, \phi_t)
\leq
\mathcal{K}(\phi_t)
\left(
\alpha(\Psi_t, \phi_t)+ 
N^{-\delta}
\right)
\end{align*}
holds
for some $\delta>0$. By a Gr\"onwall estimate, which precise form can be found below, we then obtain
$\alpha_t(\Psi_t, \phi_t) \rightarrow 0$   as $N \rightarrow \infty$, if $ \alpha(\Psi_0, \phi_0)$
converges to zero.

\begin{lemma}\label{equiv}
Let $\Psi \in L^2_s(\mathbb{R}^{2N},\mathbb{C}) \cap \mathcal{D}( (H_{W_\beta,  \cdot})^2) ,\; \| \Psi\|=1$ and let $\phi \in L^2(\mathbb{R}^2,\mathbb{C}),\; \| \phi \|=1$.
Let $\alpha(\Psi,\phi)$ be defined as above. Then,
\begin{align}
\nonumber
\lim_{N\to\infty}\alpha(\Psi,\phi)=0 \;\; \Leftrightarrow  \hspace{1cm}&\lim_{N\to\infty}\gamma^{(1)}_{\Psi}=|\phi\rangle\langle\phi|\text{ in trace norm}
\\
 \text{and}
&
\lim_{N\to\infty} \text{Var}_{ H_{W_\beta, \cdot}}(\Psi) =0
\;.
\end{align}
\end{lemma}
\begin{proof}
 $\lim_{N\to\infty}
\laa \Psi, q^{\phi}_1 \Psi \raa
=0  \Leftrightarrow  \lim_{N\to\infty}\gamma^{(1)}_{\Psi}=|\phi\rangle\langle\phi|\text{ in trace norm}$ follows from the inequality
$\laa \Psi, q^{\phi}_1 \Psi \raa
\leq
\text{Tr} | \gamma^{(1)}_{\Psi}-|\phi\rangle\langle\phi| |
\leq
 \sqrt{ 8
\laa \Psi, q^{\phi}_1 \Psi \raa}
$, see \cite{knowles}.
\end{proof}

\begin{definition}\label{alphasplit}
Let 
\begin{equation}\label{defz}\potdiff^\phi_\beta(x_j,x_k)=W_{\beta}(x_j-x_k)-\frac{a}{N-1}|\phi|^2(x_j)-\frac{a
}{N-1}|\phi|^2(x_k)
\;.
\end{equation}
Define the functional $\gamma:\LZN \times L^2(\mathbb{R}^2, \mathbb{C})
\to\mathbb{R}^+_0$ by
\begin{align}
\gamma(\Psi,\phi)
\label{ppqp}
=&2N\left|\laa\Psi ,p_{1}p_{2}\potdiff^\phi_\beta(x_1,x_2)q_{1}p_{2}
\Psi\raa\right|
\\
\label{ppqq}
+&2N\left|\laa\Psi ,p_{1}p_{2}\potdiff^\phi_\beta(x_1,x_2) q_{1}q_{2}
\Psi\raa\right|
\\
\label{qqpqq}
+&2N\left|\laa\Psi ,q_{1}p_{2}\potdiff^\phi_\beta(x_1,x_2)q_{1}q_{2}
\Psi\raa\right|\;.
\end{align}
\end{definition}

\begin{lemma}\label{ableitung}

Let $\Psi_t$ the unique solution to $i \partial_t \Psi_t
=  H_{W_\beta, t} \Psi_t$ with initial datum $\Psi_0 \in L^2_{s}(\mathbb{R}^{2N}, \mathbb{C}) \cap 
\mathcal{D}(  (H_{W_\beta, 0})^2)
,\; 
\|\Psi_0\|=1$. 
Let $\phi_t$ the unique solution to $i \partial_t \phi_t
= h^{\text{NLS}}_{a, t} \phi_t$ with initial datum $\phi_0 \in H^2(\mathbb{R}^2,\mathbb{C}) \;,\|\phi_0\|=1$.
Let $\alpha(\Psi_t,\phi_t)$ be defined as in Definition \ref{defm}. Then
\begin{align}
\label{lemmaableitungeq}
\frac{d}{dt}
 \alpha(\Psi_t,\phi_t)\leq
 \gamma(\Psi_t,\phi_t)
+
\left|
\frac{d}{dt}
\text{Var}_{ H_{W_\beta, t}} (\Psi_t)
\right|
\;.
\end{align}

\end{lemma}
\begin{remark}
The three different contributions of $\gamma(\Psi, \phi)$ can be identified with three distinct
transitions of particles out of the condensate described by $\phi$. The first line can be identified as the interaction of two particles in the state $\phi$, causing one particle to leave the condensate. The second line estimates the evaporation of two particles. The last contribution describes the interaction of one particle in the condensate with one particle outside the condensate, causing the particle in the state $\phi$ to leave the condensate.
\end{remark}

\begin{proof} For the proof of the Lemma we restore the upper index $\phi_t$ in order to pay respect to the time dependence of $p_1^{\phi_t}$ and $q_1^{\phi_t}$.
The proof is a straightforward calculation of 
\begin{align*}&
\frac{d}{dt}\laa\Psi_t,q_1^{\phi_t}\Psi_t\raa
\\=&
i\laa  H_{W_\beta, t}\Psi_t
,q_1^{\phi_t}\;\Psi_t\raa
-
i\laa\Psi_t
,q_1^{\phi_t}\; H_{W_\beta, t}\Psi_t\raa
-
i\laa \Psi_t
,
[
-\Delta_1+ a |\phi_t|^2(x_1)+ A_t(x_1),
q_1^{\phi_t}
\;]
\Psi_t\raa
\nonumber
\\=&
i(N-1)\laa\Psi_t
,[\potdiff^{\phi_t}_\beta(x_1,x_2),q^{\phi_t}_1 \;]
\Psi_t\raa
=
-2  (N-1)
\text{Im}
\left(
\laa\Psi_t
,\potdiff^{\phi_t}_\beta(x_1,x_2) q^{\phi_t}_1 
\Psi_t\raa
\right)
\;.
\end{align*}
 Using the identity
$1=p_1^{\phi_t}+ q_1^{\phi_t}$, we obtain
\begin{align*}
\frac{d}{dt}\laa\Psi_t,q_1^{\phi_t}\Psi_t\raa
=
-&
2  (N-1)
\text{Im}
\left(
\laa\Psi_t
,p^{\phi_t}_1\potdiff^{\phi_t}_\beta(x_1,x_2) q^{\phi_t}_1 
\Psi_t\raa
\right)
\\
=
-&
2  (N-1)
\text{Im}
\left(
\laa\Psi_t
,p^{\phi_t}_1 p^{\phi_t}_2\potdiff^{\phi_t}_\beta(x_1,x_2) q^{\phi_t}_1 p^{\phi_t}_2 
\Psi_t\raa
\right)
\\
-&
2  (N-1)
\text{Im}
\left(
\laa\Psi_t
,p^{\phi_t}_1 p^{\phi_t}_2\potdiff^{\phi_t}_\beta(x_1,x_2) q^{\phi_t}_1 q^{\phi_t}_2 
\Psi_t\raa
\right)
\\
-&
2  (N-1)
\text{Im}
\left(
\laa\Psi_t
,p^{\phi_t}_1 q^{\phi_t}_2 \potdiff^{\phi_t}_\beta(x_1,x_2) q^{\phi_t}_1 p^{\phi_t}_2 
\Psi_t\raa
\right)
\\
-&
2  (N-1)
\text{Im}
\left(
\laa\Psi_t
,p^{\phi_t}_1 q^{\phi_t}_2 \potdiff^{\phi_t}_\beta(x_1,x_2) q^{\phi_t}_1 q^{\phi_t}_2 
\Psi_t\raa
\right) \;.
\end{align*}
Note that $\text{Im}
\left(
\laa\Psi_t
,p^{\phi_t}_1 q^{\phi_t}_2 \potdiff^{\phi_t}_\beta(x_1,x_2) q^{\phi_t}_1 p^{\phi_t}_2 
\Psi_t\raa \right)=0$, which concludes the proof.

\end{proof}
We now establish the Gr\"onwall estimate.
\begin{lemma} \label{gammalemma}
Let $\Psi_t$ the unique solution to $i \partial_t \Psi_t
=  H_{W_\beta, t} \Psi_t$ with initial datum 
$\Psi_0 \in L^2_{s}(\mathbb{R}^{2N}, \mathbb{C})
\cap \mathcal{D} ( H_{W_\beta, 0} ^2) 
 \;, \|\Psi_0\|=1$.
 Let $\phi_t$ the unique solution to $i \partial_t \phi_t
= h^{\text{NLS}}_{a, t} \phi_t$ with initial datum $\phi_0 \in L^{2}(\mathbb{R}^2,\mathbb{C}), \;\|\phi_0\|=1$.
 Assume (A1)-(A5).
Then,
\begin{align} 
\label{theoremeq1}
\gamma(\Psi_t,\phi_t) \leq &
\mathcal{K}(\phi_t)
\ln(N)^{1/2}
\left(
\alpha(\Psi_t,\phi_t)
 + N^{-2 \beta} \ln(N)^{1/2}
 +
 N^ {-1/3} \ln(N)^{3/2}
 \right)\;.
\end{align}
\end{lemma}
The proof of this Lemma can be found in Section \ref{gammaestimates}.
\medskip\\
\begin{refproof}\ref{theo} (a):
Once we have proven Lemma \ref{gammalemma}, we obtain with Gr\"onwall's Lemma that
\begin{align}
&
\nonumber
\alpha(\Psi_t,\phi_t)\leq  
N^{\frac{\int_0^t ds \mathcal{K}(\phi_s)}{\ln(N)^{1/2}}}
\alpha(\Psi_0,\phi_0)
\\
+&
\int_0^t ds
\mathcal{K}(\phi_s)
N^{\frac{\int_s^t d \tau \mathcal{K}(\phi_\tau)}{\ln(N)^{1/2}}}
\left(
N^{-2 \beta} \ln(N)
 +
 N^ {-1/3} \ln(N)^2
 \right)
\;.
\end{align}
Note that under the assumptions (A2) and (A4)
there exists a time-dependent constant $C_t < \infty$, such that
 $\int_0^t ds
\mathcal{K}(\phi_s) \leq C_t$.
Furthermore, the assumption
$\lim\limits_{N\to\infty}\left(
N^{\delta}
\text{Tr} | \gamma^{(1)}_{\Psi_0}-|\phi_0\rangle\langle\phi_0| |
\right) =0 $ for some $\delta>0$
then implies together with (A3)
\begin{align*}
\lim_{N\to\infty}
\left(
N^{\frac{\int_0^t ds \mathcal{K}(\phi_s)}{\ln(N)^{1/2}}}
\alpha(\Psi_0,\phi_0)
\right)=0
\;,
\end{align*}
since
$\laa \Psi, q^{\phi}_1 \Psi \raa
\leq
\text{Tr} | \gamma^{(1)}_{\Psi}-|\phi\rangle\langle\phi| |
\leq
 \sqrt{ 8
\laa \Psi, q^{\phi}_1 \Psi \raa}
$, see \cite{knowles}.
Therefore,
\begin{align}
\nonumber 
\text{Tr} | \gamma^{(1)}_{\Psi_t}-|\phi_t\rangle\langle\phi_t| |
\leq
&
C
N^{\frac{C_t}{2 \ln(N)^{1/2}}- \delta/2}
\\
\label{tracestimate}
+&
C
\sqrt{
C_t
N^{\frac{\sup_{s \in [0,t]}|C_t-C_s|}{\ln(N)^{1/2}}}
\left(
N^{-2 \beta} \ln(N)
 +
 N^ {-1/3} \ln(N)^2
 \right)
 }
\;.
\end{align}

This proves  Theorem \ref{theo} (a).

 \begin{flushright} $ \Box $ \end{flushright}
\end{refproof}

\subsection{Energy estimates}

\begin{lemma}
\label{varianzlemma}
Let $\Psi_0 \in L^2_{s}(\mathbb{R}^{2N}, \mathbb{C}) \cap 
\mathcal{D}(  (H_{W_\beta, 0})^2)
$ with $\|\Psi_0\|=1$. 
Let $\Psi_t$ the unique solution to $i \partial_t \Psi_t
=  H_{W_\beta, 0} \Psi_t$ with initial datum $\Psi_0,\; \|\Psi_0 \|=1$. 
 Let $\phi_t$ the unique solution to $i \partial_t \phi_t
= h^{\text{NLS}}_{a, t} \phi_t$ with initial datum $\phi_0 \in L^{2}(\mathbb{R}^2,\mathbb{C}), \;\|\phi_0\|=1$.
Assume (A1),(A2), (A4) and (A5).
Then,
\begin{itemize}
\item[(a)]
 \begin{align}
  \| \nabla_1 \Psi_t \| \leq \mathcal{K}(\phi_t) \;.
 \end{align}
 
\item[(b)] 
\begin{align}
\|q^{\phi_t}_2 \nabla_1 \Psi_t \|^2 \leq 
\mathcal{K}(\phi_t)
\left(
\alpha (\Psi_t, \phi_t)
+
N^{-1/2} 
\right)
\;.
\end{align}

\item[(c)]
For 
any $p \in \mathbb{N}$, there exists a constant $C_p$, depending on $p$, such that
\begin{align}
&
\left\|
\sqrt{|N W_\beta (x_1-x_2)|}
 q^{\phi_t}_1 q^{\phi_t}_2  \Psi_t 
 \right\|^2
 \leq
\mathcal{K}(\phi_t)
C_p
N^{\beta/p}
\left(
\alpha (\Psi_t, \phi_t)
+
N^{-1/2} 
\right)
\;.
\end{align}
\end{itemize}

\end{lemma}
\begin{proof}
\begin{itemize}
\item[(a)]
Using (A1) together with (A2), we directly obtain
 the operator inequality
\begin{align*}
-\sum_{k=1}^N \epsilon \Delta_k
\leq &
 H_{W_\beta, t}
+CN .
\end{align*}

Using $ \frac{d}{dt}
N^{-1}
 \laa \Psi_t,  H_{W_\beta, t} \Psi_t \raa \leq
 \|\dot{A}_t \|_\infty
 $ together with (A5), the energy per particle
$N^{-1}
 \laa \Psi_t,  H_{W_\beta, t} \Psi_t \raa \leq \mathcal{K}(\phi_t) $ is uniformly bounded in $N$.
Since $\Psi_t$ is symmetric, we obtain
\begin{align*}
 N \epsilon \laa \Psi_t, -\Delta_1 \Psi_t \raa
=&
\laa \Psi_t, 
\left(
-\sum_{k=1}^N \epsilon \Delta_k
\right)
 \Psi_t \raa
\leq 
\mathcal{K}(\phi_t) N .
\end{align*}
\item[(b)] (see also \cite{leopold}.)
We estimate
\begin{align*}
&\epsilon
\| q_2 \nabla_1 \Psi_t \|^ 2
=
\frac{1}{N-1}
\laa \Psi_t, 
q_2
\epsilon
\sum_{k=1}^N(-\Delta_k) q_2 \Psi_t \raa
-
\frac{\epsilon}{N-1}
\laa \Psi_t, 
q_2
(-\Delta _2) q_2 \Psi_t \raa
\\
\leq 
&
\frac{1}{N-1}
\laa \Psi_t, 
q_2
 H_{W_\beta, t}
 q_2 \Psi_t \raa
 +
 C
\laa \Psi_t, q_1 \Psi_t \raa
\\
= &
\frac{N}{N-1}
\laa \Psi_t, 
q_2
\left(
\frac{ H_{W_\beta, t} }{N}
-
N^{-1} \laa \Psi_t,  H_{W_\beta, t} \Psi_t \raa
 \right)
 \Psi_t \raa +
\frac{1}{N-1} 
 \laa \Psi_t,  H_{W_\beta, t} \Psi_t \raa
 \laa \Psi_t, q_2 \Psi_t \raa
 \\
 -&
 \frac{1}{N-1}
\laa \Psi_t, 
q_2
 H_{W_\beta, t}
p_2
 \Psi_t \raa
  +
 C
\laa \Psi_t, q_1 \Psi_t \raa
\\
\leq &
C
\text{Var}_{ H_{W_\beta, t}} (\Psi_t)
 +
 \frac{1}{N-1}
 |
\laa \Psi_t, 
q_2
 H_{W_\beta, t}
p_2
 \Psi_t \raa
 |
 +
 2
\mathcal{K}(\phi_t)
\| q_1 \Psi_t \|^2 
.
\end{align*}
It remains to estimate
\begin{align*}
 &\frac{1}{N-1}
 |
\laa \Psi_t, 
q_2
 H_{W_\beta, t}
p_2
 \Psi_t \raa
 |
 \\
 \leq&
 \frac{1}{N-1}
|\laa \Psi_t, 
q_2
(-\Delta_2)
p_2
 \Psi_t \raa |
 +
| \laa \Psi_t, 
q_2
W_\beta (x_1-x_2)
p_2
 \Psi_t \raa |
 +
 \frac{1}{N-1}
| \laa \Psi_t,
 q_2
 A_t(x_1)
 p_2
 \Psi_t \raa |
 \\
 \leq &
\frac{1}{N-1}
\left(
\| \nabla_1 \Psi_t \|
\|\nabla \phi_t \|
+
\|\nabla \phi_t \|^2
\right)
+
\|W_\beta\| \| \phi_t \|_\infty \| \mathds{1}_{B_{CN^ {-\beta}}(0)}(x_1-x_2) \Psi_t \|
+
\|\phi_t\|_\infty^2 \|W_\beta\|_1
\\
+&
\frac{1}{N-1}
(
\|A_t^-\|_\infty
+
\|\sqrt{ A^+_t} \Psi_t \| \|\sqrt{ A_t^+} \phi_t \|
+
\langle \phi_t, A_t \phi_t \rangle
),
\end{align*}
where we used $q_2=1-p_2$ for all three contributions in the last inequality.
Recall the two-dimensional Sobolev's inequality as presented in e.g.  \cite{liebanalysis}, Theorem 8.5.
For any $\rho \in H^1(\mathbb{R}^2, \mathbb{C})$ and for any $2 \leq p < \infty$, there exists a constant $C_p$, depending on $p$, such that
\begin{align}
\| \rho \|_p^2
\leq
C_p
\left(
\| \rho \|^2 + \|\nabla \rho \|^2
\right)
\end{align}
holds.
It is shown in Theorem 8.5. in \cite{liebanalysis} that $C_p$ fulfills
$C_p \leq C p$.
We use this inequality in the $x_1$ variable and obtain together with H\"older's inequality, to obtain
\begin{align*} 
&
\| \mathds{1}_{B_{CN^ {-\beta}}(0)}(x_1-x_2) \Psi_t \|^2
\\
\leq &
\|  \mathds{1}_{B_{CN^ {-\beta}}(0)}\|_{ \frac{N}{N-1}} 
\int d^2 x_2\dots d^2x_N 
\left(
\int d^2x_1  |\Psi_t(x_1, \dots,x_N)|^{2N}
\right)^{1/N}
 \\
 \leq &
 C
 N^{1- 2 \beta}
 \int d^2 x_2\dots d^2x_N 
\left(
\int d^2x_1  |\nabla_1 \Psi_t(x_1, \dots,x_N)|^{2}
+
\int d^2 x_1  | \Psi_t(x_1, \dots,x_N)|^{2}
\right)
\;.
\\
\leq &
 C
  N^{1- 2 \beta}
 \left(
 \| \nabla_1\Psi_t \|^2
 +
  \| \Psi_t \|^2
 \right).
 \end{align*}
With $\|W_\beta\|  = CN^{-1 + \beta} $, we obtain together with (a)
\begin{align}
\|W_\beta\|  \| \mathds{1}_{B_{CN^ {-\beta}}(0)}(x_1-x_2) \Psi_t \|
\leq
\mathcal{K}(\phi_t)
N^{-1/2}.
\end{align}

Next, we show that $\|\sqrt{ A_t^+} \Psi_t \|$ and $ \|\sqrt{ A_t^+} \phi_t \|$ are uniformly
bounded in $N$. 
Using
 the operator inequality (A1) together with (A2) and (A5)
directly implies
\begin{align*}
\epsilon 
\laa \Psi_t
\sum_{k=1}^N A^+_t(x_k)
\Psi_t \raa
\leq 
 \laa \Psi_t,  H_{W_\beta, t} \Psi_t \raa +\mathcal{K}(\phi_t) N
\leq 
\mathcal{K}(\phi_t) N
\;.
\end{align*}
To control $\langle \phi_t, A^+_t \phi_t \rangle$,
let $\Omega_t =\phi_t^{\otimes N}$. Then
\begin{align*}
&\epsilon 
\langle \phi_t, A^+_t
\phi_t \rangle
\leq
N^{-1} \laa \Omega_t,  H_{W_\beta, t} \Omega_t \raa
+
 \mathcal{K}(\varphi_t)
\\
=&
\langle \phi_t, 
\left(
-\Delta 
+
\frac{a}{2} |\phi_t|^2
+
 A_t 
 \right)
 \phi_t \rangle
 +
  \mathcal{K}(\varphi_t)
 \\
 +&
 \langle \phi_t, 
\left(
\frac{1}{2}
 (N-1) W_\beta \ast |\phi_t|^2
-
\frac{a}{2} |\phi_t|^2
 \right)
 \phi_t \rangle
\;.
\end{align*}
Note that
\begin{align*}
&
\left|
\frac{d}{dt}
\langle \phi_t, 
\left(
-\Delta 
+
\frac{a}{2}
|\phi_t|^2 
+
 A_t 
 \right)
 \phi_t \rangle
 \right|
\leq
\| \dot{A}_t \|_\infty \;.
\end{align*}
which implies together with (A5)
\begin{align*}
\langle \phi_t, 
\left(
-\Delta 
+
\frac{a}{2}
|\phi_t|^2 
+
 A_t 
 \right)
 \phi_t \rangle
 \leq \mathcal{K}(\phi_t) \;.
\end{align*}
Furthermore,  we obtain as in \eqref{pppqcancelation}
\begin{align*}
& 
\left|
\langle \phi_t, 
\left(
 (N-1) W_\beta \ast |\phi_t|^2
-
a|\phi_t|^2
 \right)
 \phi_t \rangle
 \right|
\\
\leq &
\| \phi_t\|_\infty^2
\left(
\left\|
N W_\beta \ast |\phi_t|^2
-
a | \phi|^2
 \right\|
 +
 \|W_\beta\|_1 \|\phi_t\|_\infty^2
 \right)
 \\
\leq&
 \mathcal{K}(\varphi_t)
(
 N^{-2\beta} \ln(N) + N^{-1})
\;.
\end{align*}
This concludes the proof of (b).

\item[(c)]
First, we like to recall the following Gagliardo-Nirenberg interpolation inequality:
for any $\varrho \in H^1(\mathbb{R}^2,\mathbb{C})$
any for any $2<m<\infty$, there exists a constant $C_m$, depending only on $m$, such that
$\| \varrho \|_m \leq
 C_m \| \nabla \varrho\|^{\frac{m-2}{m}} \| \varrho \|^{\frac{2}{m}}
$ holds. 
For $1<p<\infty$, we 
estimate, using H\"older's - and the inequality above
\begin{align}
&
\nonumber
\left\|
\sqrt{|N W_\beta (x_1-x_2)|}
 q_1 q_2  \Psi_t 
 \right\|^2
 \leq
\| N W_\beta\|_\infty
\left\|
\mathds{1}_{B_{CN^{-\beta}}(0)}(x_1-x_2)
 q_1 q_2  \Psi_t 
 \right\|^2
 \\
 \nonumber
 \leq &
 C N^{2 \beta}
 \int_{\mathbb{R}^{2N-2}}d^2x_2 \dots d^2x_N
 \int_{\mathbb{R}^2} d^2x_1
 | (q_1 q_2  \Psi_t) (x_1, \dots, x_N)|^2
 \mathds{1}_{B_{CN^{-\beta}}(0)}(x_1-x_2)
 \\
 \nonumber
 \leq &
  C_p N^{2 \beta}
\|  \mathds{1}_{B_{CN^{-\beta}}(0)}\|_{\frac{p}{p-1}}
 \int_{\mathbb{R}^{2N-2}}d^2x_2 \dots d^2x_N
\left(
 \int_{\mathbb{R}^2} d^2x_1
 | (q_1 q_2  \Psi_t) (x_1, \dots, x_N)|^{2p}
\right)^{1/p}
\\
\nonumber
 \leq &
  C_p N^{ 2\beta
\left(1- \frac{p-1}{p}   \right)
  }
 \int_{\mathbb{R}^{2N-2}}d^2x_2 \dots d^2x_N
\left(
 \int_{\mathbb{R}^2} d^2x_1
 | (\nabla_1 q_1 q_2  \Psi_t) (x_1, \dots, x_N)|^{2}
\right)^{\frac{p-1}{p}}
\\
\label{labelneu}
\times &
\left(
 \int_{\mathbb{R}^2} d^2 \tilde{ x}_1
 | (q_1 q_2  \Psi_t) (\tilde{ x}_1, \dots, x_N)|^{2}
\right)^{1/p}
\end{align}
We use H\"older's inequality with respect to the $x_2, \dots x_N$-integration with the conjugate pair
$r = \frac{p}{p-1}$ and $s=p$ to obtain
\begin{align*}
\eqref{labelneu}
\leq
  C_p N^{ \frac{2\beta}{p}  
}
\| \nabla_1 q_1q_2 \Psi_t\|^ {2 \frac{p-1}{p}}
\| q_1 q_2 \Psi_t\|^{\frac{2}{p}}.
\end{align*}
Note that
\begin{align*}
\|\nabla_1 q_1 q_2 \Psi_t \|^2
\leq
2
\| \nabla_1 p_1 q_2 \Psi_t \|^2
+
2
\|\nabla_1 q_2 \Psi_t \|^2
\leq
\mathcal{K}(\phi_t)
\left(
\alpha( \Psi_t, \phi_t)
+
N^{-1/2}
\right).
\end{align*}
Renaming $p$, we thus obtain with part (b), that there exists a constant depending on $p$ such that
\begin{align*}
\left\|
\sqrt{|N W_\beta (x_1-x_2)|}
 q_1 q_2  \Psi_t 
 \right\|^2
 \leq
 C_p
 \mathcal{K}(\phi_t)
 N^{\beta/p}
\left(
 \alpha(\Psi_t, \phi_t)+N^{-1/2}\right) .
\end{align*}
\end{itemize}

\end{proof}

\subsection{Proof of Lemma \ref{gammalemma} }
\label{gammaestimates}

Following a general strategy, we will smear out the potential $W_\beta$ in order to use smoothness properties of $\Psi$. For this, we define
\begin{definition}\label{udef}%(Smoothing out the interaction)
For any $0\leq\beta_1\leq  \beta$, we define
\begin{align}
U_{\beta_1}(x)=\left\{
         \begin{array}{ll}
         \frac{a}{\pi}
		N^{-1+2\beta_1}  & \hbox{for \ $|x|<N^{-\beta_1}$,} \\
           0 & \hbox{else.}
         \end{array}
       \right.
\end{align}

and
\begin{align} \label{defh} h_{\beta_1,\beta}(x)=  \frac{1}{2\pi} \int_{\mathbb R^2}   \ln|x-y| 
(W_{\beta}(y)-U_{\beta_1}(y))d^2y
\; .
\end{align}
\end{definition}
\begin{lemma}\label{ulemma}
For any $0\leq  \beta_1\leq  \beta$, we obtain with the above definition
\begin{enumerate}
\item
\begin{align}
&\Delta h_{\beta_1,\beta}=W_{\beta}-U_{\beta_1}
.
\end{align}
\item
 \begin{align}
 \|h_{\beta_1,\beta}\| \leq &C N^{-1-\beta_1}\ln(N) \text{ for } \beta_1>0
 \; ,
 \\
\|h_{0,\beta}\| \leq &C N^{-1} \text{ for } \beta>0
\; ,
\\
 \|\nabla h_{\beta_1,\beta}\|\leq& CN^{-1} (\ln (N))^{1/2} \;,
\end{align}
 \end{enumerate}
\end{lemma}

\begin{proof} 
The lemma has been proven in \cite{jeblick}, Lemma 7.2. for nonnegative potentials $W_\beta$. It is easy to verify that this proof  also works without specifying the sign of $W_\beta$. We thus refer the reader to
\cite{jeblick} for the details of the proof.

\end{proof}
Next, we prove some well known properties of the projectors $p_j,q_j$.
\begin{lemma}\label{kombinatorik}
 Let $f\in L^1\left(\mathbb{R}^2,\mathbb{C}\right)$, $g\in L^2\left(\mathbb{R}^2,\mathbb{C}\right)$. Then,
\begin{align}\label{kombeqa}\|p_j f(x_j-x_k)p_j\|_{\text{op}}%\leq  \|f\star|D\phi|^2\|_\infty
\leq&  \|f\|_1\|\phi\|_\infty^2\;,
\\ \label{kombeqb}
\|p_jg^*(x_j-x_k)\|_{\text{op}}=&
\|g(x_j-x_k)p_j\|_{\text{op}}\leq  \|g\|\;\|\phi\|_\infty
%\;,
\\\label{kombeqc}
\| |\phi(x_j) \rangle \langle \nabla_j \phi(x_j)| h^* (x_j-x_k)\|_{\text{op}}=&
\|h(x_j-x_k)\nabla_j p_j\|_{\text{op}}\leq  \|h\|\|\nabla\phi\|_{\infty}
\;.\end{align}

\end{lemma}
\begin{proof}
First note that, for bounded operators $A, B$,
$\| A B \|_{\text{op}}=\|  B^* A^* \|_{\text{op}}$ holds, where $A^*$ is the adjoint operator of $A$.
To show (\ref{kombeqa}), note that
\begin{align} p_j f(x_j-x_k)p_j=p_j (f\star|\phi|^2)(x_k)\;.
\label{faltungorigin}
\end{align}
It follows that $$\|p_j f(x_j-x_k)p_j\|_{\text{op}}\leq  \|f\|_1\|\phi\|_\infty^2\;.$$

For (\ref{kombeqb}) we write
\begin{align*}
\|g(x_j-x_k)p_j\|_{\text{op}}^2=&\sup_{\|\Psi\|=1}\|g(x_j-x_k)p_j\Psi\|^2=
\\=&\sup_{\|\Psi\|=1}\laa \Psi,p_j  |g(x_j-x_k)|^2 p_j\Psi\raa\\\leq& \|p_j  |g(x_j-x_k)|^2p_j\|_{\text{op}}\;.
\end{align*}
With (\ref{kombeqa}) we get (\ref{kombeqb}).
For \eqref{kombeqc} we use
\begin{align*}
\| g(x_j-x_k) \nabla_j p_j \|_{\text{op}}^ 2 
=&
\sup_{\|\Psi\|=1}
\laa \Psi,
p_j
(|g|^ 2 * |\nabla \phi|^2)(x_k)
\Psi \raa
\leq
\||g|^ 2 * |\nabla \phi|^2\|_{\infty}
 \\
 \leq
 &
 \|g\|^2 \| \nabla \phi \|_\infty^2
\end{align*}
\end{proof}
We further need the following

\begin{lemma}\label{trick}
Let $\Omega,\chi\in L^ 2(\mathbb{R}^{2N}, \mathbb{C}) $ 
such that $\Omega$ and $\chi$ are symmetric w.r.t. to the exchange of the variables $x_2, \dots, x_N$.
Let $O_{i,j}$ be an operator acting on the $i^{th}$ and $j^{th}$  coordinate. Then
\begin{align}
|\laa\Omega,O_{1,2}\chi\raa|&\leq \|\Omega\|^2+\left|\laa O_{1,2}\chi,O_{1,3}\chi\raa\right|+
\frac{1}{N-1}\|O_{1,2}\chi\|^2
\;.
\end{align}
\end{lemma}
\begin{proof}
Using symmetry and Cauchy Schwarz
\begin{align*}
|\laa\Omega,O_{1,2}\chi\raa|=
&\frac{1}{N-1}|\laa\Omega,\sum_{j=2}^N O_{1,j}\chi\raa|\leq
\frac{1}{N-1}\|\Omega\|\;
\left\|\sum_{j=2}^N O_{1,j}\chi \right\|
\\ \leq &
\|\Omega\|^2
+
\frac{1}{(N-1)^2}
\left\|\sum_{j=2}^N O_{1,j}\chi \right\|^2
\;.
\end{align*}
Note that
\begin{align*}
\left\|\sum_{j=2}^N O_{1,j}\chi \right\|^2&=\laa\sum_{j=2}^N O_{1,j}\chi,\sum_{k=2}^N O_{1,k}\chi\raa
\\\leq&\sum_{j=2}^N 
|\laa O_{1,j}\chi,O_{1,j}\chi\raa|+
|\sum_{j\neq k,j,k=2}^N \laa O_{1,j}\chi,O_{1,k}\chi\raa|
\\\leq&(N-1)|\laa O_{1,2}\chi,O_{1,2}\chi\raa|+(N-1)(N-2)|\laa O_{1,2}\chi,O_{1,3}\chi\raa|
\;.
\end{align*}
\end{proof}

%%%%%%%%%%%%%%%%%%%%%%%%%%%%

We now prove Lemma \ref{gammalemma}.
\begin{lemma}\label{hnorms}
Let $\Psi  \in L^2_{s}(\mathbb{R}^{2N}, \mathbb{C}) \cap \mathcal{D}( H_{W_\beta, \cdot} ), \| \Psi \|=1$
and let $\phi \in L^2(\mathbb{R}^{2}, \mathbb{C}) ,\; \| \phi \|=1$.
Assume (A1), (A2) and (A5).
 Then,
\begin{enumerate}
\item[(a)]  
\begin{align}
N\left|\laa\Psi p_1p_2
\potdiff^\phi_\beta(x_1,x_2) q_1p_2\Psi\raa\right|\leq   
\mathcal{K}(\phi)
 \left( N^{-1}+   N^{- 2 \beta } \ln(N) \right)
\;.
\end{align}
\item[(b)]
 \begin{align} 
  N|\laa\Psi, p_1p_2
\potdiff^\phi_\beta(x_1,x_2)
 q_1q_2\Psi\raa|
 \leq 
\mathcal{K}(\phi)
	\left(
		\laa \Psi, q_1 \Psi \raa 
		 +
N^{-1/3} \ln(N)^2 
\right) 
\;.
\end{align}
\item[(c)]
\begin{align} &
N\left|\laa\Psi,  q_1 p_2
\potdiff^\phi_\beta(x_1,x_2)
 q_1 q_2\Psi\raa\right|
 \leq 
 \mathcal{K}(\phi)
 \ln(N)^{1/2}
 \left(
\alpha(\Psi, \phi)
+ 
N^{-1/2 }
\right) 
\;.
\end{align}
\item[(d)]
Let $\phi_t$  the solution to $ i \partial_t \phi_t= h^{\text{NLS}}_{a, t} \phi_t, \; \| \phi_0 \|=1$.
Let $\Psi_t$ the solution to $i \partial \Psi_t= H_{W_\beta,t} \Psi_t$ with
$\Psi_0 \in L^2_{s}(\mathbb{R}^{2N}, \mathbb{C}) \cap \mathcal{D}( (H_{W_\beta, 0})^ 2 ), \| \Psi_0 \|=1$
 Then,
\begin{align}
\left|
\frac{d}{dt}
\text{Var}_{ H_{W_\beta, t}} (\Psi_t)
\right|
\leq
\mathcal{K}(\phi_t)
\left(
\alpha(\Psi_t, \phi_t)
+
 N^{-1}
\right) \;.
\end{align}

\end{enumerate}

\end{lemma}
\begin{remark}
(a) and (b) have essentially been proven in \cite{jeblick} for a slightly different definition of $\alpha (\Psi, \phi)$. It is (c) where our old estimates fail. In \cite{jeblick} we we able to control (c) by estimating $\| \nabla_1 q_1 \Psi\|$ in terms of $\alpha(\Psi, \phi)$ and some small error. In this estimate, it was crucial that
$
(1-p_1p_2) W_\beta(x_1-x_2) (1- p_1p_2) \geq 0
$ holds as an operator inequality and therefore forces $W_\beta$ to be nonnegative
(cf. the definition of $Q_\beta (\Psi, \phi)$ below Equation (100) in \cite{jeblick}).
In this paper, we make use of a strategy which was developed in \cite{leopold} to derive the Maxwell-Schr\"odinger equations from the Pauli-Fierz Hamiltonian. Instead of estimating 
$
\| \nabla_1 q_1 \Psi \|
$, we now control $\| \nabla_1 q_2 \Psi \|$, see Lemma \ref{varianzlemma}.
\end{remark}
\begin{proof}

\begin{enumerate}
\item
We estimate
\begin{align*}
N \left|\laa\Psi, p_1p_2
\potdiff^\phi_\beta(x_1,x_2) q_1p_2\Psi\raa\right|
\leq  N\|p_1p_2
\potdiff^\phi_\beta(x_1,x_2) q_1p_2\|_{\text{op}}\;.
\end{align*}
$\|p_1p_2
\potdiff^\phi_\beta(x_1,x_2) q_1p_2\|_{\text{op}}$ can be estimated using $p_1q_1=0$ and (\ref{faltungorigin}).
\begin{align*}
& N
\left\|
	p_1p_2
	\left(
		W_\beta(x_1-x_2)-\frac{a}{N-1}|\phi(x_1)|^2-
		\frac{a}{N-1}|		
		\phi(x_2)|^2
	\right) q_1p_2
\right\|_{\text{op}}
\\
&\leq\|p_1p_2 
(NW_\beta(x_1-x_2)-a |\phi(x_1)|^2)
p_2\|_{\text{op}}
+C \|\phi\|_\infty^2 N^{-1}
\\
&\leq \| \phi \|_{\infty} \ \|N(W_\beta\star|\phi|^2)-a |\phi|^2\|+C \|\phi\|_\infty^2 N^{-1}\;.
\end{align*}
Let $h$ be given by 
$$h(x)= \frac{1}{2 \pi}\int_{\mathbb{R}^2}d^2y \ln |x-y| NW_\beta(y) - \frac{a}{2\pi}\ln |x|
\;.
 $$
It then follows
$$\Delta h(x)=N W_\beta(x) - a \delta(x)$$ 
in the sense of distributions.
 Since $a = \int_{\mathbb{R}^2}d^2x W(x)$, this implies  (see Lemma \ref{ulemma}), 
$h(x)=0$ for $x \notin B_{RN^{- \beta}}(0)$, where $R N^{-\beta}$ is the radius of the support of $W_\beta$.
Thus,
\begin{align*}
\|h\|_1
\leq&
\frac{1}{2 \pi}
\int_{\mathbb{R}^2}d^2x \int_{\mathbb{R}^2}  d^2y|\ln |x-y|| \mathds{1}_{B_{RN^{- \beta}}(0)}(x) NW_\beta(y)
\\
+&
\frac{|a|}{2 \pi}
\int_{\mathbb{R}^2}  d^2x  \ln(|x|) \mathds{1}_{B_{RN^{- \beta}}(0)}(x)
\\
\leq &
C N^{-2 \beta} \ln(N) \;.
\end{align*}
Integration by parts and Young's inequality then imply
\begin{align}
\label{pppqcancelation}
&\| N(W_\beta\star|\phi|^2)- a|\phi|^2\|=\|(\Delta h)\star|\phi|^2\|
\\
\nonumber
\leq& \| h\|_1 \|\Delta|\phi|^2\|
\leq \mathcal{K}(\varphi) N^{-2\beta} \ln(N) 
\;.
\end{align}
Thus, we obtain the bound
\begin{align}
N \left|\laa\Psi, p_1p_2
\potdiff^\phi_\beta(x_1,x_2) q_1p_2\widehat{w}\Psi\raa\right|
\leq
\mathcal{K}(\phi)
 \left( N^{-1}+   N^{- 2 \beta } \ln(N) \right)
\;,
\end{align}
which then proves (a).

\item
We will first consider the case $0<\beta \leq 1/3$.
Note that
$
p_1 p_2 \potdiff^\phi_\beta(x_1,x_2)q_1 q_2
=
p_1 p_2 W_\beta(x_1-x_2)q_1 q_2
$.
We estimate
\begin{align*}
&N |\laa\Psi ,q_1 q_2 W_\beta(x_1-x_2) p_1 p_2
\Psi\raa|
=
\frac{N}{N-1}
|\laa q_1\Psi , 
\sum_{k=2}^N
q_k W_\beta(x_1-x_k) p_1 p_k
\Psi\raa|
\\
\leq &
\frac{N}{N-1}
\|q_1\Psi  \|
\left\|
\sum_{k=2}^N
q_k W_\beta(x_1-x_k) p_1 p_k
\Psi
\right\|
\\
\leq &
\laa \Psi, q_1 \Psi \raa
+
\left\|
\sum_{k=2}^N
q_k W_\beta(x_1-x_k) p_1 p_k
\Psi
\right\|^2
\\
=&
\laa \Psi, q_1 \Psi \raa
+
(N-1)
\left\|
q_2 W_\beta(x_1-x_2) p_1 p_2
\Psi
\right\|^2
\\
+&
(N-1)(N-2)
\laa \Psi, 
p_1 p_2 W_\beta(x_1-x_2) q_2
q_3 W_\beta(x_1-x_3) p_1 p_3
\Psi \raa
\\
\leq &
\laa \Psi, q_1 \Psi \raa
+
N
\|W_\beta \|^2 \|\phi \|_\infty^2
\\
+&
N^2
\| p_2 W_\beta(x_1-x_2) p_1 \|_{\text{op}}^2
\| q_1 \Psi \|^2
\\
\leq &
\mathcal{K}(\phi)
\left(
\laa \Psi, q_1 \Psi \raa
+
N^{-1+ 2 \beta}
\right) \;.
\end{align*}
In the last estimate, we used Lemma \ref{ulemma} together with Lemma \ref{kombinatorik} to estimate
$
\| p_1 W_\beta(x_1-x_2) p_2 \|_{\text{op}}
\leq
\|p_1 \sqrt{|W_\beta|} \|_{\text{op}}^2
\leq
\|\phi\|_\infty^2 \|\sqrt{|W_\beta|}\|^2
\leq
\mathcal{K}(\phi) N^{-1}
$.

This proves (b) for the case $\beta \leq 1/3$.

\item[(b)] for $1/3< \beta$:
We use $U_{\beta_1}$ from Definition \ref{udef} for some $0<\beta_1 \leq 1/3$.
 We then obtain
 \begin{align}
 \nonumber
&   N\laa\Psi, p_1p_2
 W_{\beta}(x_1-x_2)
q_1q_2\Psi\raa
\\ 
\label{eins}
  =&
    N\laa\Psi, p_1p_2
U_{\beta_1}(x_1-x_2)
q_1q_2\Psi\raa
  \\
  \label{zwei}
  +&
    N\laa\Psi, p_1p_2
  \left(W_{\beta}(x_1-x_2)-U_{\beta_1}(x_1-x_2)\right) q_1q_2\Psi\raa.
 \end{align}
Term \eqref{eins} has been controlled above.
So we are left to control \eqref{zwei}.

Let  $\Delta h_{\beta_1,\beta}=W_{\beta}-U_{\beta_1}$, as in Lemma \ref{ulemma}. Integrating by parts and using $\nabla_1 h_{\beta_1,\beta}(x_1-x_2)=-\nabla_2 h_{\beta_1,\beta}(x_1-x_2)$ gives
 \begin{align}
\nonumber
&N\left|\laa\Psi, p_1p_2\left(W_{\beta}(x_1-x_2)-U_{\beta_1,\beta}(x_1-x_2)\right)q_1q_2\Psi\raa\right|
\\\label{csplit1}&
\leq N\left|\laa\nabla_1p_1\Psi, p_2\nabla_2 h_{\beta_1,\beta}(x_1-x_2) q_1q_2\Psi\raa\right|
\\\label{csplit2}
&+N\left|\laa \Psi, p_1p_2\nabla_2 h_{\beta_1,\beta}(x_1-x_2) \nabla_1q_1q_2\Psi\raa\right|
\;.
\end{align}
Let $ t_1 \in \lbrace p_1, \nabla_1 p_1 \rbrace$ and let
 $\Gamma \in \lbrace q_1\Psi, \nabla_1 q_1\Psi \rbrace$.

For both \eqref{csplit1} and \eqref{csplit2}, we use 
 Lemma \ref{trick} with 
 $O_{1,2}=N^{1+ \eta/2} q_2\nabla_2 h_{\beta_1,\beta}(x_1-x_2)p_2$, $\chi=t_1\Psi$ and
 $\Omega= N^{- \eta/2} \Gamma$. This yields
 
\begin{align}
\label{gammaterm}
&\eqref{csplit1}+\eqref{csplit2}
\leq 
2 
\sup_{t_1 \in  \lbrace p_1, \nabla_1 p_1 \rbrace,  \Gamma \in \lbrace q_1\Psi, \nabla_1 q_1\Psi \rbrace }
\Big(
N^{-\eta} \| \Gamma \|^2
\\
\label{diagonalterm}
+&
\frac{N^{2+\eta}}{N-1}\|q_2\nabla_2 h_{\beta_1,\beta}(x_1-x_2)t_1p_2\Psi\|^2
 \\
\label{nebendiagonale}
+&
N^{2+\eta}\left|\laa \Psi,t_1p_2 q_3 \nabla_2 h_{\beta_1,\beta}(x_1-x_2)\nabla_3 h_{\beta_1,\beta}
(x_1-x_3)t_1q_2 p_3\Psi\raa\right|
\Big)
\;.
 \end{align}
The first term can be bounded using $\|\nabla_1 q_1 \Psi\| \leq \mathcal{K}(\phi)$.

Using $\|t_1\Psi\|^2 \leq \mathcal{K}(\phi)$, we obtain
\begin{align*}
 (\ref{diagonalterm})
 \leq&
 \mathcal{K}(\phi)
 \frac{N^{2+\eta}}{N-1}\|\nabla_2 h_{\beta_1,\beta}(x_1-x_2)p_2\|_{\text{op}}^2
\leq 
 \mathcal{K}(\phi)
\frac{N^{2+\eta}}{N-1}
\|\phi\|_\infty ^2
\|\nabla h_{\beta_1,\beta}\|^2
\\
\leq&
 \mathcal{K}(\phi)
N^{\eta-1}\ln(N)
 \;,
 \end{align*}
where we used Lemma \ref{ulemma} in the last step.

Next, we estimate 
\begin{align*}
(\ref{nebendiagonale})\leq& N^{2+\eta}\| p_2
\nabla_2 h_{\beta_1,\beta}(x_1-x_2)t_1 q_2
\Psi\|^2\
\nonumber
\\
\leq&
2N^{2+\eta}\| p_2
 h_{\beta_1,\beta}(x_1-x_2)t_1 \nabla_2q_2
\Psi\|^2
\nonumber
\\+&
2N^{2+\eta}\||\varphi(x_2)\rangle \langle \nabla \varphi(x_2)|
 h_{\beta_1,\beta}(x_1-x_2)t_1 q_2
\Psi\|^2\
\nonumber
\\
\leq &
2N^{2+\eta}\| p_2
 h_{\beta_1,\beta}(x_1-x_2)
 \|_{\text{op}}^2
 \|
t_1 
  \nabla_2 q_2
\Psi\|^2\
\nonumber
\\
+&
2N^{2+\eta}
\| |\varphi(x_2)\rangle \langle \nabla \varphi(x_2)|
 h_{\beta_1,\beta}(x_1-x_2) \|_{\text{op}}^2
  \| t_1 q_2
\Psi\|^2\
\nonumber
\\
\leq &
 \mathcal{K}(\phi)
 N^{2+\eta}
\| h_{\beta_1,\beta} \|^2
\nonumber
\leq
 \mathcal{K}(\phi)
N^{\eta- 2\beta_1} \ln(N)^2
\;.
\end{align*}
Thus, for all $\eta \in \mathbb{R}$
\begin{align*}
&  N\laa\Psi, p_1p_2
  \left(W_{\beta}(x_1-x_2)-U_{\beta_1,\beta}(x_1-x_2)\right)
 q_1q_2\Psi\raa
  \nonumber
\\ 
  \leq &
  \mathcal{K}(\phi,A_t)
  \left(  N^{ - \eta} + N^{ \eta-1} \ln(N) +N^{ \eta- 2 \beta_1} \ln(N)^2 \right).
\end{align*}
Hence, we obtain,
using $ N^{ \eta-1} \ln(N) < N^{ \eta- 2 \beta_1} \ln(N)$, 
 \begin{align*} 
 & N\laa\Psi, p_1p_2
  W_{\beta}(x_1-x_2)
q_1q_2\Psi\raa
  \nonumber
\\ 
 & \leq 
\mathcal{K}(\phi,A_t)
	\left(
		\laa \Psi, q_1 \Psi \raa 
		 +
		  \inf_{\eta>0}
  \inf_{\frac{1}{3}\geq \beta_1 >0}
  \left(
 			N^{ \eta- 2 \beta_1} \ln(N)^2
 				 +N^{-1+2 \beta_1}
+  N^{ - \eta}
\right)  
\right) 
   \;.
\end{align*}
 and we get (b)  in full generality by choosing $\eta=\beta_1=1/3$.

\item[(c)]
First note that
\begin{align*}
N \left|
\laa \Psi,
q_1 p_2 \frac{a}{N-1} |\phi(x_1)|^2 q_1q_2 \Psi \raa 
\right|
\leq
C \| \phi \|_\infty^2
\laa \Psi, q_1 \Psi \raa \;.
\end{align*}
Let $U_{0}$ be given by Definition \ref{udef}. Using Lemma \ref{kombinatorik} and integrating by parts we get
\begin{align}\nonumber
N&\left|\laa\Psi,   q_1p_2W_\beta(x_1-x_2) q_1q_2\Psi\raa\right|
\nonumber 
\\\leq&\nonumber
N\left|\laa\Psi,q_1  p_2U_{0}(x_1-x_2) q_1q_2 \Psi\raa\right|+
N\left|\laa\Psi,q_1  p_2(\Delta_1
h_{0,\beta}(x_1-x_2)) q_1q_2\Psi\raa\right|
\\\leq&\label{term1}
N\|q_1\Psi\|\;\|U_{0}\|_\infty\|q_1q_2\Psi\|
\\%
\label{term2}
+&N\left|\laa \nabla_2 p_2  q_1   \Psi,(\nabla_2
h_{0,\beta}(x_1-x_2))q_1q_2\Psi\raa\right|
\\
\label{term3}
+&
N\left|\laa\Psi, q_1p_2(\nabla_2
h_{0,\beta}(x_1-x_2))
\nabla_2   q_1q_2\Psi\raa\right|
\;.
\end{align}
The first contribution is bounded by
$$\eqref{term1}
\leq
C\laa\Psi,q_1\Psi\raa
\;.
$$
The second term $\eqref{term2}$ can be estimated as 
\begin{align}
\eqref{term2} 
\label{noideaforalabel}
=
&N\left|\laa  \Delta_2 p_2 q_1  \Psi,
h_{0,\beta}(x_1-x_2)
q_1q_2\Psi\raa\right|
\\ 
\label{noideaforalabel2}
+&
N\left|\laa   \nabla_2 p_2 q_1  \Psi,
h_{0,\beta}(x_1-x_2)
q_1 \nabla_2
\Psi\raa\right|
\\
+& \label{noideaforalabel3}
N\left|\laa  \nabla_2 p_2 q_1\Psi,
h_{0,\beta}(x_1-x_2)
q_1 \nabla_2p_2
\Psi\raa\right|   .
\end{align}
The last contribution $\eqref{term3}$ can be rewritten as
\begin{align}
\eqref{term3} 
\leq &
\label{noideaforalabel4}
N\left|\laa\Psi, q_1p_2(\nabla_2
h_{0,\beta}(x_1-x_2))
\nabla_2 q_1 \Psi\raa\right|
\\
+
\label{noideaforalabel5}
&N\left|\laa\Psi, q_1p_2(\nabla_2
h_{0,\beta}(x_1-x_2))
\nabla_2p_2  q_1  \Psi\raa\right|
\;.
\end{align}
We estimate each contribution separately, using 
Lemma \ref{ulemma} together with Lemma \ref{kombinatorik}. We obtain
\begin{align*}
\eqref{noideaforalabel}
\leq &
N \|q_1  \Psi \|
\| 
h_{0,\beta}(x_1-x_2) \Delta_2 p_2 \|_{\text{op}}
\|  q_1 q_2  \Psi \|
\\
\leq & 
C \| \Delta\phi \|_\infty
\laa \Psi, q_1 \Psi \raa
 .
\end{align*}
Analogously,
\begin{align*}
\eqref{noideaforalabel3}
\leq &
N \|q_1  \Psi \|
\| h_{0,\beta}(x_1-x_2)  \nabla_2 p_2 \|_{\text{op}}
\| \nabla \phi \|_\infty
\|q_1 \Psi \|
\\
\leq &
C \| \nabla\phi \|_\infty^2
\laa \Psi, q_1 \Psi \raa
 .
\end{align*}
Next, we control
\begin{align*}
\eqref{noideaforalabel5}
\leq &
N \|q_1  \Psi \|
\|  p_2 \nabla_2 h_{0,\beta}(x_1-x_2) 
\nabla_2 p_2
\|_{\text{op}}
\|q_1 \Psi \| 
\\
\leq &
C 
\|\nabla \phi \|_\infty
\|\phi\|_\infty N \|\nabla h_{0,\beta}\|_1
 \|q_1 \Psi \|^2 
\\
\leq &
C  \|\nabla \phi \|_\infty \|\phi\|_\infty
\laa \Psi, q_1 \Psi \raa .
\end{align*}
To control \eqref{noideaforalabel2} and \eqref{noideaforalabel4}, 
we estimate for $t_2 \in \lbrace p_2, |\phi(x_2) \rangle\langle(\nabla \phi)(x_2)| \rbrace$ and
$s \in \lbrace
h_{0,\beta},
\nabla_2
h_{0,\beta} \rbrace$
\begin{align*}
&
N\left|\laa q_1\Psi,   t_2
s(x_1-x_2)
\nabla_2 q_1 \Psi\raa\right|
\\
\leq  &
\ln(N)^{1/2}
\| q_1 \Psi \|^2
+
\ln(N)^{-1/2}
N^2 \|t_2 s(x_1-x_2)\|_{\text{op}}^2 
\| \nabla_2 q_1 \Psi \|^2
\\
\leq &
\mathcal{K}(\phi)
\ln(N)^{1/2}
\left(
\text{Var}_{ H_{W_\beta, \cdot}} (\Psi)
+
\laa \Psi, q_1 \Psi \raa
+
N^{-1/2}
\right) .
\end{align*}
In the last line, we used Lemma \ref{varianzlemma} (b) together with Lemma \ref{ulemma}.

\item[(d)]
We estimate
\begin{align*}
&
\left|
\frac{d}{dt}
\text{Var}_{ H_{W_\beta, t}} (\Psi_t)
\right|
=
N^{-2}
\left|
\frac{d}{dt}
\laa \Psi_t,
\left(
 H_{W_\beta, t}
-
\laa \Psi_t,  H_{W_\beta, t} \Psi_t \raa
\right)^2
\Psi_t \raa
\right|
\\
\leq&
2 
N^{-1}
\left|
\laa \Psi_t,
\left(
 H_{W_\beta, t}
-
\laa \Psi_t,  H_{W_\beta, t} \Psi_t \raa
\right)
\left(
\dot{A}_t(x_1)
-
\laa \Psi_t,
\dot{A}_t(x_1)\Psi_t \raa
\right)
\Psi_t \raa
\right|
\\
\leq&
2 N^{-1}
\left|
\laa \Psi_t,
\left(
 H_{W_\beta, t}
-
\laa \Psi_t,  H_{W_\beta, t} \Psi_t \raa
\right)
\left(
p_1\dot{A}_t(x_1)p_1
-
\laa \Psi_t,
p_1\dot{A}_t(x_1)p_1\Psi_t \raa
\right)
\Psi_t \raa
\right|
\\
+&
2 N^{-1}
 \left|
\laa \Psi_t,
\left(
 H_{W_\beta, t}
-
\laa \Psi_t,  H_{W_\beta, t}\Psi_t \raa
\right)
\left(
p_1\dot{A}_t(x_1)q_1
-
\laa \Psi_t,
p_1\dot{A}_t(x_1)q_1
\Psi_t \raa
\right)
\Psi_t \raa
\right|
\\
+&
2 N^{-1}
 \left|
\laa \Psi_t,
\left(
 H_{W_\beta, t}
-
\laa \Psi_t,  H_{W_\beta, t} \Psi_t \raa
\right)
\left(
N^{-1}
\sum_{k=1}^N
q_k\dot{A}_t(x_k)p_k
-
\laa \Psi_t,
q_1\dot{A}_t(x_1)p_1
\Psi_t \raa
\right)
\Psi_t \raa
\right|
\\
+&
2  N^{-1}
\left|
\laa \Psi_t,
\left(
 H_{W_\beta, t}
-
\laa \Psi_t,  H_{W_\beta, t} \Psi_t \raa
\right)
\left(
q_1\dot{A}_t(x_k)q_1
-
\laa \Psi_t,
q_1\dot{A}_t(x_k)q_1\Psi_t \raa
\right)
\Psi_t \raa
\right|
\\
\leq &
4\text{Var}_{ H_{W_\beta, t}} (\Psi_t)
+
\left\|
\left(
p_1\dot{A}_t(x_1)p_1
-
\laa \Psi_t, 
p_1\dot{A}_t(x_1)p_1\Psi_t \raa
\right)
\Psi_t
\right\|^2
\\
+&
\left\|
\left(
p_1\dot{A}_t(x_1)q_1
-
\laa \Psi_t, 
p_1\dot{A}_t(x_1)q_1\Psi_t \raa
\right)
\Psi_t
\right\|^2
\\
+&
\left\|
\left(
N^{-1}
\sum_{k=1}^N
q_k\dot{A}_t(x_k)p_k
-
\laa \Psi_t, 
q_1\dot{A}_t(x_1)p_1\Psi_t \raa
\right)
\Psi_t
\right\|^2
\\
+&
\left\|
\left(
q_1\dot{A}_t(x_1)q_1
-
\laa \Psi_t, 
q_1\dot{A}_t(x_1)q_1\Psi_t \raa
\right)
\Psi_t
\right\|^2
\end{align*}
Note that
\begin{align*}
&\left\|
\left(
p_1\dot{A}_t(x_1)p_1
-
\laa \Psi_t, 
p_1\dot{A}_t(x_1)p_1\Psi_t \raa
\right)
\Psi_t
\right\|^2
\\
=&
\left(
\int_{\mathbb{R}^2}d^2x
\dot{A}_t(x) |\phi_t(x)|^2
\right)^2
\laa \Psi_t, p_1 \Psi_t \raa
\left(
1- \laa \Psi_t, p_1 \Psi_t \raa
\right)
\\
\leq  &
\mathcal{K}(\phi_t) \laa \Psi_t, q_1 \Psi_t \raa \;.
\end{align*}
Furthermore
\begin{align*}
&\left\|
\left(
N^{-1}
\sum_{k=1}^N
q_k\dot{A}_t(x_k)p_k
-
\laa \Psi_t, 
q_1\dot{A}_t(x_1)p_1\Psi_t \raa
\right)
\Psi_t
\right\|^2
\\
=&
N^{-2}
\sum_{k,l=1}^N
\laa \Psi_t, p_l\dot{A}_t(x_l)q_l q_k\dot{A}_t(x_k)p_k
\Psi_t \raa
-
\left|
\laa \Psi_t, 
q_1\dot{A}_t(x_1)p_1\Psi_t \raa
\right|^2
\\
\leq &
\laa \Psi_t, p_1\dot{A}_t(x_1)q_1 q_2\dot{A}_t(x_2)p_2
\Psi_t \raa
+
\frac{1}{N}
\laa \Psi_t, p_1\dot{A}_t(x_1)q_1 \dot{A}_t(x_1)p_1
\Psi_t \raa
+ \|\dot{A}_t \|_{\infty} \laa \Psi_t, q_1 \Psi_t \raa
\\
\leq &
\mathcal{K}(\phi_t)
\left(
\laa \Psi_t, q_1 \Psi_t \raa
+
\frac{1}{N}
\right)
\;.
\end{align*}
To control the two remaining terms, let $s_1 \in \lbrace p_1, q_1 \rbrace$. Then, we need to estimate
\begin{align*}
&\left\|
\left(
s_1\dot{A}_t(x_1)q_1
-
\laa \Psi_t, 
s_1\dot{A}_t(x_1)q_1\Psi_t \raa
\right)
\Psi_t
\right\|^2
\\
=&
\laa \Psi_t, 
q_1 \dot{A}_t(x_1)
s_1\dot{A}_t(x_1)q_1\Psi_t \raa
-
\left|
\laa \Psi_t, 
s_1\dot{A}_t(x_1)q_1\Psi_t \raa
\right|^2
\\
\leq &
2 
\| \dot{A}_t \|_\infty^2
\laa \Psi_t, q_1 \Psi_t \raa
\;.
\end{align*}
In total, we obtain
\begin{align*}
\left|
\frac{d}{dt}
\text{Var}_{ H_{W_\beta, t}} (\Psi_t)
\right|
\leq
\mathcal{K}(\phi_t)
\left(
\alpha (\Psi_t, \phi_t)
+
 N^{-1}
\right) \;.
\end{align*}

\end{enumerate}
Combining the estimates (a)-(d), Lemma \ref{gammalemma} is then proven.

\end{proof}

\section{Proof of Theorem \ref{theo} (b)} \label{sobolevtracenorm}

\begin{proof}
We make use of the inequality
\begin{align*}
\text{Tr}
\left|
\sqrt{1-\Delta}
( \gamma_{\Psi_t}^{(1)}- | \phi_t \rangle \langle \phi_t |)
\sqrt{1-\Delta}
\right|
\leq
C
(1+ \| \nabla_1 \phi_t\|)^2
(
\|q_1 \Psi_t\|
+
\|q_1 \Psi_t\|^2
+
\|\nabla_1 q_1 \Psi_t\|
+
\|\nabla_1 q_1 \Psi_t\|^2
),
\end{align*}
which was proven in \cite{picklnorm}, see also \cite{anapolitanos}.
Using Theorem \ref{theo} (a), we are left to show $
\lim\limits_{N \rightarrow \infty}\| \nabla_1 q_1 \Psi_t\|=0$. 
In general, this does not follow from 
$\lim\limits_{N \rightarrow \infty}\| \nabla_1 q_2 \Psi_t\|=0$.
To see this, consider the symmetrized wave-function
\begin{align*}
\Gamma (x_1, \dots, x_N)
=
\frac{1}{\sqrt{ N}} \sum_{k=1}^N
\phi (x_1) \dots \eta(x_k)\dots \phi(x_N)
\end{align*}
for $\eta, \phi \in L^2(\mathbb{R}^{2}, \mathbb{C}),\; \|\eta\|= \|\phi\|=1,\; \langle \eta, \phi \rangle= 0$. Then 
\begin{align*}
&\| q^{\phi}_1 \Gamma\|^2= N^{-1}, \qquad
\| \nabla_1 q^{\phi}_2 \Gamma\|^2 =N^{-1} \|\nabla \phi\|^2
, \qquad
\| \nabla_1 q^{\phi}_1 \Gamma\|^2= N^{-1} \|\nabla \eta\|^2 \;.
\end{align*}
Note that $\|\nabla \eta\|$ can be chosen arbitrarily.
However, for  $A_t \in L^p(\mathbb{R}^2, \mathbb{R})$, with $p>2$,
it is possible to control $\| \nabla_1 q_1 \Psi_t\|$ in terms of
$\| q_1 \Psi_t \|$, $\| \nabla_1 q_2 \Psi_t\|$
and the energy difference $ \left|
 N^ {-1}\laa \Psi_0,  H_{W_\beta, 0} \Psi_0 \raa 
 -
 \langle \phi_0, 
\left(
-\Delta 
+
\frac{a}{2}
 |\phi_0|^2
+
 A_0 
 \right)
 \phi_0 \rangle 
 \right|$, assuming conditions (A2) and (A4). 
 Together with assumptions (A1), (A3), (A5) and (A5)' and Theorem \ref{theo}, part (a), it is then possible to
 bound $\| \nabla_1 q_1 \Psi_t\|$ sufficiently well.
First, we consider 
\begin{align}
|\|\nabla_1 \Psi_t \|^ 2- \| \nabla \phi_t \|^ 2|
\leq&
\label{energydifference}
\left|
\frac{1}{N}
\laa \Psi_0, H_{W_\beta, 0 } \Psi_0 \raa
-
\langle\phi_0, \left(
-\Delta 
+
\frac{a}{2}
 |\phi_0|^2
+
 A_0 
 \right)
  \phi_0 \rangle
\right|
\\
+&
\label{energychange}
\int_0^ t ds
\left|
\laa \Psi_s, \dot{A}_s (x_1) \Psi_s \raa
-
\langle \phi_s, \dot{A}_s \phi_s \rangle
\right|
\\
\label{ppwpp-mf}
+&
\frac{1}{2}
\left|
\laa \Psi_t, 
p_1 p_2
(N-1)
W_\beta (x_1-x_2)
p_1 p_2
 \Psi_t \raa
 -
 a
 \langle \phi_t, |\phi_t|^2 \phi_t \rangle
 \right|
\\
\label{firsterrror}
 +&
N
\left|
 \laa \Psi_t, 
p_1 p_2
W_\beta (x_1-x_2)
(1-p_1 p_2)
 \Psi_t \raa
 \right|
 \\
 \label{seconderrror}
 +&
 N
 \left|
  \laa \Psi_t, 
(1-p_1 p_2)
W_\beta (x_1-x_2)
(1-p_1 p_2)
 \Psi_t \raa 
\right| 
 \\
 +&
 \label{A-Mf}
 \left|
 \laa \Psi_t, A_t(x_1) \Psi_t \raa
 -
 \langle \phi_t, A_t \phi_t \rangle
 \right|
 .
\end{align}
We estimate each line separately. From condition (A5)', it follows that
$\eqref{energydifference} \leq C N^ {-\delta}$.
Using $\dot{A}_t \in L^\infty (\mathbb{R}^2, \mathbb{R})$, we estimate
\begin{align*}
\eqref{energychange}
\leq &
t
\sup_{s \in [0, t]}
\left(
|\langle \phi_s, \dot{A}_s \phi_s \rangle | \laa \Psi_s, q^{\phi_s}_1 \Psi_s \raa
+
2
 \left|\laa \Psi_s,p^{\phi_s}_1 \dot{A}_s(x_1)q^{\phi_s}_1 \Psi_s \raa \right|
 +
  \left|\laa \Psi_s,q^{\phi_s}_1 \dot{A}_s(x_1)q^{\phi_s}_1 \Psi_s \raa \right|
  \right)
\\
\leq &
t
\sup_{s \in [0, t]}
\left(
\|\dot{A}_s\|_\infty (\|q^{\phi_s}_1 \Psi_s\|+ \|q^{\phi_s}_1 \Psi_s\|^ 2)
\right).
\end{align*}

 Next,
\begin{align*}
\eqref{ppwpp-mf}
 \leq &
 \left|
  \langle \phi_t, (N-1)W_\beta \star |\phi_t|^2 \phi_t \rangle
\laa \Psi_t, 
p_1 p_2
 \Psi_t \raa
 -
 a
 \langle \phi_t, |\phi_t|^2 \phi_t \rangle
 \right|
 \\
 \leq &
 \mathcal{K}(\phi_t)
 \|N W_\beta\|_1 \|q_1 \Psi_t \|^2
 +
 \left|
   \langle \phi_t, 
   (
 (N-1)  W_\beta \star |\phi_t|^2 
-a | \phi_t|^2   
   )
   \phi_t \rangle
 \right|
\\
\leq &
\mathcal{K}(\phi_t)
\left(
\laa \Psi_t, q_1 \Psi_t \raa
+
N^{-2 \beta} \ln(N)+ N^{-1} 
 \right)
 .
\end{align*}

Note that
\begin{align*}
& 
\eqref{firsterrror}
+
\eqref{seconderrror}
\leq
C
 \left\|
\sqrt{N |W_\beta(x_1-x_2)|}
 p_1 p_2  \Psi_t 
 \right\|
 \left(
 \left\|
\sqrt{N |W_\beta (x_1-x_2)|}
 q_1 q_2  \Psi_t 
 \right\|
 +
 \left\|
\sqrt{N |W_\beta(x_1-x_2)|}
 q_1 p_2  \Psi_t 
 \right\|
 \right)
 \\
 &
 +C
 \left\|
\sqrt{N |W_\beta(x_1-x_2)|}
 p_1 q_2  \Psi_t 
 \right\|^2 
 +C
 \left\|
 \sqrt{N |W_\beta(x_1-x_2)|}
 q_1 q_2  \Psi_t 
 \right\|^2
 .
\end{align*}
Using Lemma \ref{varianzlemma} (c), we
 obtain 
\begin{align*}
\left\|
\sqrt{|N W_\beta (x_1-x_2)|}
 q_1 q_2  \Psi_t 
 \right\|^2
 \leq
 \mathcal{K}(\phi_t)
 C_p
 N^{\beta/p}
\left(
 \alpha(\Psi_t, \phi_t)+
 N^{-1/2 }\right) .
\end{align*}

Furthermore,
\begin{align*}
\left\|\sqrt{N |W_\beta (x_1-x_2)|}
 p_1
 \right\|_{\text{op}}
 \leq
 \mathcal{K}(\phi_t).
\end{align*}
Note that
it was shown in part (a) that $
 \alpha(\Psi_t, \phi_t)+
 \leq
  \mathcal{K}(\phi_t)
N^{-\delta}$ for some $\delta >0$. 
Choosing $p \in \mathbb{N}$ large enough,
we then obtain
$
\eqref{firsterrror}
+
\eqref{seconderrror}
\leq
  \mathcal{K}(\phi_t) N^{-\gamma}
$, for some $\gamma>0$.
We estimate, using $|\laa \Psi_t, q_1 A_t(x_1) q_1 \Psi_t\raa|
\leq
\|q_1 \Psi_t\|( \|A_t \phi_t\|+ \|  A_t(x_1) \Psi_t\|)$,
\begin{align*}
|\eqref{A-Mf}|
\leq &
|\langle \phi_t, A_t \phi_t \rangle | \laa \Psi_t, q_1 \Psi_t \raa
+
2
 \left|\laa \Psi_t,p_1 A_t(x_1)q_1 \Psi_t \raa \right|
 +
  \left|\laa \Psi_t,q_1 A_t(x_1)q_1 \Psi_t \raa \right|
\\
\leq &
|\langle \phi_t, A_t \phi_t \rangle | \laa \Psi_t, q_1 \Psi_t \raa
+
\|q_1 \Psi_t\|( \|A_t \phi_t\|+ \|  A_t(x_1) \Psi_t\|).
\end{align*}
If $A_t \in L^\infty (\mathbb{R}^2, \mathbb{R})$ holds, we obtain
\begin{align*}
|\eqref{A-Mf}| \leq
 \mathcal{K}(\phi_t)
 (\|q_1 \Psi_t \|+ \|q_1 \Psi_t\|^2).
\end{align*}
On the other hand,
using Sobolev and H\"older inequality (see the proof of Lemma \ref{varianzlemma}), 
together with $|\langle \phi_t, A_t \phi_t \rangle |+
\|\nabla_1 \Psi_t\|+ \|\nabla \phi_t\| \leq \mathcal{K}(\phi_t)
$, we obtain, for any $1<p<\infty$ 
\begin{align*}
|\eqref{A-Mf}| \leq
 \mathcal{K}(\phi_t)
 \left(
 1+
\|A_t\|_{\frac{2p}{p-1}}
\right)
 (\|q_1 \Psi_t \|+ \|q_1 \Psi_t\|^2).
\end{align*}

Therefore, if $A_t \in L^p(\mathbb{R}^2, \mathbb{C})$ holds for some
$
p \in ]2, \infty]
$ and for all $t \in \mathbb{R}$, we obtain
\begin{align*}
|\|\nabla_1 \Psi_t \|^ 2- \| \nabla \phi_t \|^ 2|
\leq
t
\sup_{s \in [0,t]}
\left(
	\mathcal{K}(\phi_s) 
	\left(
		\alpha (\Psi_s, \phi_s)+
		\sqrt{\alpha (\Psi_s, \phi_s)}+ N^{-\delta}+N^{-1}
	\right)
\right)
.
\end{align*}
Since
\begin{align*}
\| \nabla_1 q_1 \Psi_t\|^2
\leq
|\|\nabla_1 \Psi_t \|^ 2- \| \nabla \phi_t \|^ 2|
+
\| \nabla \phi_t\|^2 \laa \Psi_t, q_1 \Psi_t \raa
+
2 \|\nabla\phi_t \|  \|q_1 \Psi_t\| \|\Psi_t\|
\end{align*}
holds, we obtain with part (a) of Theorem \ref{theo}, part (b) of Theorem \ref{theo}.

\end{proof}

\section{Appendix}
\subsection{Energy variance of a product state}
\begin{lemma} \label{varianceproduct}
Let
$\Psi = \phi^{\otimes N}$ and assume that
$
\|\phi\|_\infty+ \|-\Delta\phi\|_\infty+ 
\|-\Delta \phi\|
+
\| \nabla \phi\|
+ \langle \phi, A_s \phi \rangle
+
\langle \phi, A_s^2 \phi \rangle \leq C
$.
Then, 
\begin{align}
\text{Var}_{ H_{W_\beta, s}}(\Psi)
\leq
C(N^{-1}+N^{-1+\beta}+N^{-2+ 2 \beta})
\;.
\end{align}

\end{lemma}

\begin{proof}
The proof is a direct calculation using the product structure of $\Psi= \phi^{\otimes N}$.
We first calculate,
denoting $T=\sum_{k=1}^N (-\Delta_k) $, $\mathcal{W}= \sum_{i<j}^N W_\beta (x_i-x_j)$ and
$\mathcal{A}= \sum_{k=1}^N A_s (x_k)$,

\begin{align*}
\frac{1}{N^2}
\laa \Psi,  H_{W_\beta, s} \Psi \raa^2=&
\frac{1}{N^2}
\laa
\Psi,
\left(
T+\mathcal{W}+\mathcal{A}
\right)
\Psi
\raa^2
\\
=&
\frac{1}{N^2}
(
N \langle \phi, -\Delta \phi \rangle
+
\frac{N (N-1)}{2}\langle \phi, W_\beta \ast |\phi|^2 \phi \rangle
+
N \langle \phi, A_s \phi \rangle
)^2
\\
=&
 \langle \phi, -\Delta \phi \rangle^2
+
\frac{(N-1)^2}{4}\langle \phi, W_\beta \ast |\phi|^2 \phi \rangle^2
+
 \langle \phi, A_s \phi \rangle^2
 \\
 +&
 (N-1) 
\langle \phi, -\Delta \phi \rangle
\langle \phi, W_\beta \ast |\phi|^2 \phi \rangle
+
2 \langle \phi, -\Delta \phi \rangle \langle \phi, A_s \phi \rangle
\\
+&
 (N-1)
 \langle \phi, A_s \phi \rangle 
 \langle \phi, W_\beta \ast |\phi|^2 \phi \rangle
.
\end{align*}
It then follows that
\begin{align}
\nonumber
\text{Var}_{ H_{W_\beta, s}}(\Psi)=&
 \laa \Psi, 
\frac{ H_{W_\beta, s}^2 }{N^2} 
\Psi \raa
-
\frac{1}{N^2}
\laa \Psi,  H_{W_\beta, s} \Psi \raa^2
 \\
 \label{variance1}
=&
\frac{1}{N^2}
\laa \Psi,
T^2
\Psi \raa
-
 \langle \phi, -\Delta \phi \rangle^2
 \\
  \label{variance2}
+ &
\frac{1}{N^2}
2 \text{Re}
(\laa T \Psi, \mathcal{W} \Psi \raa )
-
 (N-1) 
\langle \phi, -\Delta \phi \rangle
\langle \phi, W_\beta \ast |\phi|^2 \phi \rangle
\\
\label{variance3}
+&
\frac{1}{N^2}
\laa \Psi, \mathcal{W}^2 \Psi \raa
-
\frac{(N-1)^2}{4}\langle \phi, W_\beta \ast |\phi|^2 \phi \rangle^2
\\
\label{variance4}
+&
\frac{1}{N^2}
\laa \Psi,
\mathcal{A}^2
\Psi \raa
-
 \langle \phi, A_s \phi \rangle^2
\\
\label{variance5}
+&
\frac{1}{N^2}
2 \text{Re}
(\laa \mathcal{A} \Psi, T \Psi \raa )
-
2 \langle \phi, -\Delta \phi \rangle \langle \phi, A_s \phi \rangle
\\
\label{variance6}
+&
\frac{1}{N^2}
2 \text{Re}
(\laa  \mathcal{A} \Psi,  \mathcal{W} \Psi \raa )
-
 (N-1)\langle \phi, W_\beta \ast |\phi|^2 \phi \rangle
 \langle \phi, A_s \phi \rangle .
\end{align}
We estimate each line separately.
\begin{align*}
|\eqref{variance1}|
=&
\left|
\frac{1}{N}
\laa \Psi,
(-\Delta_1)^ 2
\Psi \raa
+
\frac{N-1}{N}
\laa \Psi,
(-\Delta_1)(-\Delta_2)
\Psi \raa
-
\langle \phi, -\Delta \phi \rangle^2
\right|
\\
\leq &
\frac{
\| - \Delta \phi \|^2+ \|\nabla \phi \|^4
}{N}
\;.
\end{align*}
Note that
\begin{align*}
&
\frac{1}{N^2}
2 \text{Re}
(\laa T \Psi, \mathcal{W} \Psi \raa )
\\
=
&
\frac{1}{N^2}
\sum_{k=1}^N
\sum_{i \neq j =1}^N
\text{Re}
(\laa (-\Delta_k) \Psi, W_\beta(x_i-x_j) \Psi \raa )
\\
=&
\frac{2(N-1)}{N}
\text{Re}
(\laa (-\Delta_1) \Psi, W_\beta(x_1-x_2) \Psi \raa )
\\
+&
\frac{(N-1)(N-2)}{N}
\text{Re}
(\laa (-\Delta_1) \Psi, W_\beta(x_2-x_3) \Psi \raa )
\\
\leq &
\frac{2(N-1)}{N}
\| \Delta \phi \|_{\infty}
\| W_\beta(x_1-x_2)\|_1 \|\phi\|_\infty
\\
+&
\frac{(N-1)(N-2)}{N}
\|\nabla \phi \|^2 \langle \phi , W_\beta \ast |\phi|^2 \phi \rangle
\\
\leq &
C N^{-1}
\| \Delta \phi \|_\infty \|\phi\|_\infty 
+
\frac{(N-1)(N-2)}{N}
\langle \phi, -\Delta \phi \rangle \langle \phi , W_\beta \ast |\phi|^2 \phi \rangle ,
\end{align*}
which immediately implies
\begin{align*}
|\eqref{variance2}|
\leq
C 
\frac{\| \nabla \phi \|_\infty \| \phi\|_\infty+ 
\| \nabla \phi \|^2 \| \phi \|_\infty^4	
}{N }
\;.
\end{align*}

Next, we calculate
\begin{align*}
\frac{1}{N^2}
\laa \Psi, \mathcal{W}^2 \Psi \raa
=&
\frac{1}{4 N^2}
\sum_{i \neq j =1}^N
\sum_{k \neq l =1}^N
\laa \Psi, W_\beta (x_i-x_j) W_\beta (x_k-x_l) \Psi \raa
\\
=&
\frac{N-1}{2 N}
\laa \Psi, W_\beta (x_1-x_2) ^2 \Psi \raa
\\
+&
\frac{(N-1)(N-2)}{N}
\laa \Psi, W_\beta (x_1-x_2) W_\beta (x_2-x_3)  \Psi \raa
\\
+&
\frac{(N-1)(N-2)(N-3)}{4N}
\laa \Psi, W_\beta (x_1-x_2)  W_\beta (x_3-x_4) \Psi \raa
\;.
\end{align*}
The first term is bounded by
\begin{align*}
\frac{N-1}{2 N}
\laa \Psi, W_\beta (x_1-x_2) ^2 \Psi \raa
\leq
\|\phi \|^2_\infty \|W_\beta \|^2 \leq C N^{-2+2 \beta}   \|\phi \|^2_\infty
\;.
\end{align*}
The second term can be bounded using
\begin{align*}
f(x_2)=&
N^{-1+2 \beta}
\left| \int_{\mathbb{R}^2}dx_1| \phi (x_1)|^2 W(N^\beta(x_1-x_2)) \right|
\\
\leq &
N^{-1} \int_{\mathbb{R}^2} dx_1 | \phi (N^{-\beta} x_1) |^2 | W(x_1-N^\beta x_2) |
\leq
N^{-1} \|W \|_1 \|\phi \|^2_\infty
\end{align*}
by
\begin{align*}
&\frac{(N-1)(N-2)}{N}
\laa \Psi, W_\beta (x_1-x_2) W_\beta (x_2-x_3) \Psi \raa
\\
=&
\frac{(N-1)(N-2)}{N}
\int_{\mathbb{R}^2} dx_2 
|\phi (x_2)|^2 f(x_2)^2
\leq 
\frac{1}{N}
\|W\|_1^2
\|\phi \|^4_\infty .
\end{align*}
It therefore follows that
\begin{align*}
|\eqref{variance3}|
\leq
C
N^{-2+ 2 \beta} \| \phi \|_\infty^2
+
C N^{-1}
( \| \phi\|_\infty^ 2
+
\| \phi\|_\infty^ 4
)
\;.
\end{align*}
\eqref{variance4} is estimated by
\begin{align*}
|\eqref{variance4}|
=&
\left|
\frac{1}{N}
\laa \Psi,
A_s(x_1)^2
\Psi \raa
+
\frac{N-1}{N}
\laa \Psi,
A_s(x_1) A_s(x_2)
\Psi \raa
-
\langle \phi, A_s \phi \rangle^2
\right|
\\
\leq &
\frac{
\langle \phi, A_s^2 \phi \rangle+ \langle \phi, A_s \phi \rangle^2
}{N}
\;.
\end{align*}
Furthermore,
\begin{align*}
|\eqref{variance5}|
\leq&
\left|
\frac{2}{N}
|\laa \Psi, A_s(x_1) (-\Delta_1) \Psi \raa |
+
2 
\frac{N-1}{N}
\laa \Psi,
A_s(x_1) (-\Delta_2)
\Psi \raa
-
2 \langle \phi, -\Delta \phi \rangle
\langle \phi, A_s \phi \rangle
\right|
\\
\leq &
C \frac{\|-\Delta \phi\| \|A_s \phi\|+ \|\nabla \phi\|^2 \langle \phi, A_s \phi \rangle}{N}
\;.
\end{align*}
Finally,
\begin{align*}
|\eqref{variance6}|
\leq&
C
\left(
| \laa A_s(x_1) \Psi, W_\beta(x_1-x_2) \Psi \raa |
+
| \laa \Psi, A_s(x_1) W_\beta(x_2-x_3) \Psi \raa|
\right)
\\
\leq &
C
(
\| A_s \phi \| \|W_\beta\| \|\phi\|_\infty
+
\langle \phi, A_s \phi \rangle 
\|W_\beta\|_1
\|\phi\|_\infty^2 
)
\\
\leq &
C
(
N^{-1+\beta}
\| A_s \phi \|  \|\phi\|_\infty
+
N^{-1}
\langle \phi, A_s \phi \rangle 
\|\phi\|_\infty^2 
) \;.
\end{align*}
\end{proof}

\subsection{Persistence of regularity of $\phi_t$}
\label{solutionTheory}
We study the nonlinear Schr\"odinger equation in two spatial dimensions \eqref{nls} with a harmonic potential
\begin{align*}
i \partial_t \phi_t=
(-\Delta + a | \phi_t|^2+|x|^2) \phi_t 
\end{align*}
under the conditions $a>-a^*$ and $\| \phi_0\|=1$. The solution theory of \eqref{nls} is well studied in absence of external fields. There, the global existence and persistence of regularity of $\phi_t \in H^k(\mathbb{R}^2, \mathbb{C})$ was established, assuming $\phi_0$ regular enough \cite{cazenave}. 
The condition $a>-a^*$ is known to be optimal, that is, for $a<-a^*$, there exist blow-up solutions.
It is interesting to note that global existence of solutions in
$L^\infty(\mathbb{R}^2, \mathbb{C})$ directly implies  persistence of higher regularity of solutions in $H^k(\mathbb{R}^2, \mathbb{C})$, see \cite{cazenave} and below. 
\begin{lemma}
Let $\phi_0 \in H^1(\mathbb{R}^2, \mathbb{C}), \|\phi_0 \|=1$ such that
$
\| \nabla \phi_0\|^2+
\||x| \phi_0 \|^2
+
\frac{a}{2}
\langle \phi_0, |\phi_0|^2 \phi_0 \rangle
\leq C
$.
Let $a>-a^*$.
\begin{enumerate}
\item
The nonlinear Schr\"odinger equation
\begin{align*}
i \partial_t \phi_t
=
(
-\Delta
+
a
|\phi_t|^2
+
|x|^2
) 
\phi_t
\end{align*}
admits a solution $\phi_t \in H^1(\mathbb{R}^2, \mathbb{C})$ globally in time. 
\item
Define the norm 
$
\| u \|_{\Sigma,m}
=
\sqrt{
\sum_{k=0}^m
(
\| \nabla^k u \|^2+ \| |x|^k u\|^2
)
}
$. Then
\begin{align*}
\| \phi_t \|_{\Sigma,4}
\leq
\| \phi_0 \|_{\Sigma,4}
e^{ C \int_0^t ds \| \phi_s \|_\infty^2} \;.
\end{align*}
\item
Assume $\| \phi_0  \|_{\Sigma,4}< \infty$. Then,
 there exist a time dependent constant $C_t$, also depending on
$\|\phi_0 \|_{\Sigma,4}$, such that
$\| \phi_t \|_{\Sigma,4} \leq C_t$.

\end{enumerate}

\end{lemma}
\begin{remark}
Part (c) directly implies that $\phi_t \in H^4(\mathbb{R}^2,\mathbb{C})$.
Our proof relies on the works of \cite{ carles2, carles, cazenave, killip, tao, weinstein}, see also the references therein. It also might be possible to show a polynomial growth in $t$ of the constant $C_t$, using the refined estimates presented in \cite{ carles2, carles}.
\end{remark}

\begin{proof}
\begin{enumerate}
\item
The global existence in $H^1(\mathbb{R}^2,\mathbb{C})$ is well known, see Remark 3.6.4 in \cite{cazenave}. We sketch the proof for completeness.
Let $U_t$ denote the generator of the time evolution of the
linear Schr\"odinger equation $i \partial_t u_t= (-\Delta+|x|^2) u_t$. 
For any $\phi_0 \in L^2(\mathbb{R}^2, \mathbb{C})$, we consider the Duhamel formula
\begin{align}
\label{Nlsintegral}
\phi_t=
U_t \phi_0
-i a
\int_{0}^t ds
U_{t-s}  | \phi_s|^2 \phi_s \;.
\end{align}
Note that is is  known  that 
there exists a nonempty open interval $I$, $0 \in I$ such that 
\eqref{Nlsintegral}  has a unique solution 
$\phi_t$, provided the initial datum $\phi_0$ fulfills $\|\phi_0\|_{\Sigma,1}\leq C$ 
(see Proposition 1.5. in \cite{carles}).
Furthermore, for any $t \in I$,
$\|\phi_t \|= \|\phi_0\|=1$. We may assume that $I$ is the maximal interval on which 
a solution of \eqref{Nlsintegral} exists.
Assume now that $\phi_t$ blows up in finite time, i.e. $I$ is bounded. It is then known  
that $
\int_{0}^{\sup I} dt
\|\phi_t\|_4^4= \infty 
$ \cite{killip}.

Assume $t \in I$ and consider the NLS energy
\begin{align*}
\mathcal{E}_{\text{NLS}}(\phi_t)=
\| \nabla \phi_t \|^2
+
 \frac{a}{2}
 \langle \phi_t, 
 |\phi_t|^2
 \phi_t \rangle
 +
 \| |x| \phi_t \|^2  \;.
\end{align*} 
Under the conditions $a>-a^*, \|\phi_0\|=1$, the two dimensional Gagliardo-Nirenberg inequality
$
\frac{a^*}{2} \|u\|_4^4 \leq
\| \nabla u \|^2 \|u\|^2
$, $u \in H^1(\mathbb{R}^2, \mathbb{C})$
implies  that
$
\mathcal{E}_{\text{NLS}}(\phi_t)
> 0
$.  Furthermore
$\frac{d}{dt}
\mathcal{E}_{\text{NLS}}(\phi_t)=0$, see Proposition 1.6. in \cite{carles}.
This directly implies that there exists an $\epsilon>0$ such that
\begin{align*}
 \epsilon \| \nabla \phi_t\|^2
 \leq
 C
 \;.
\end{align*}
The two dimensional 
Gagliardo-Nirenberg inequality implies, together with $\|\phi_t\|= \|\phi_0\| \,, \forall t \in I$,
\begin{align*}
\int_{0}^{\sup I}
dt
\|\phi_t\|_4^4
\leq
C
\int_{0}^{\sup I}
dt
\|\nabla \phi_t \|^2
\leq C \sup I < \infty \;.
\end{align*}
Therefore, the solution $\phi_t$ of \eqref{Nlsintegral} exists globally in time and fulfills
$\phi_t \in H^1(\mathbb{R}^2, \mathbb{C}), \| |x| \phi_t\|< \infty$.
\item
Let $A(x)=|x|^2$ and
define, for any $u \in L^2(\mathbb{R}^2,\mathbb{C})$, the norm
 $\| u \|_{k,A}  = 
\sqrt{
\sum_{m=0}^k
 \|(-\Delta+A)^m u \|^2
 }$. Note that  $\| \cdot \|_{k,A}$ is invariant under $U_t$, that is
$\|U_t u\|_{k,A}= \|u\|_{k,A}$. We will first show that
$ \| u \|_{2,A} $ and $ \| u \|_{\Sigma,4}$ are equivalent norms.
Let $u \in H^4(\mathbb{R}^2, \mathbb{C})$.
Note that
\begin{align*}
\| u \|_{2,A}^2
=&
\|u\|^2
+
\|(-\Delta+A) u \|^2
+
\|(-\Delta+A)^2 u \|^2
\\
\leq&
\|u\|^2
+
2 \|-\Delta u \|^2 +2 \| A u \|^2
\\
+&
\|
\left(
(-\Delta)^2 +A^2+ (-\Delta A)+ 2 A(-\Delta)
-2(\nabla A)\cdot \nabla 
\right)   
    u \|^2 
\\
\leq&
C \big(
\|u\|^2
+
\|-\Delta u \|^2 + \| A u \|^2
+
\|(-\Delta)^2 u \|^2
+
 \| A^2 u \|^2
 \\
+&
\| A(-\Delta) u \|^2
+
\|(\nabla A)\cdot \nabla u \|^2
\big)
\,.
\end{align*}
Since $
\nabla A^2 =
4 |x|^2 x, \;
\Delta A^2 =12 A
$, we obtain,
\begin{align*}
\| A(-\Delta) u \|^2
=&
\langle u, (-\Delta) A^2 (-\Delta) u  \rangle
\\
=&
\langle u, (-\Delta A^2) (-\Delta) u  \rangle
+
2
\langle u, (-\nabla A^2)\cdot \nabla  (-\Delta) u  \rangle
+
\langle u,  A^2 (-\Delta)^2  u  \rangle
\\
\leq &
C
(
 \| A  u \| \|-\Delta u \|
+
\| |x|^3 u \| \| \nabla \Delta u \|
+
\|A^2 u \| \| (-\Delta)^2 u \|
)
\\
\leq 
&
C (
\|A u \|^2+ \|-\Delta u \|^2 + \|(-\Delta)^2 u \|^2
+
\|A^2 u \|^2
) \;.
\end{align*}
For the last inequality, we used $
\| |x|^3 u \|^2
=
\langle |x|^2 u, |x|^4 u \rangle
\leq
\| A u\|^2+ \|A^2 u \|^2$, as well as
$
\| \nabla \Delta u \|
\leq
\| -\Delta u \|^2+ \|(-\Delta)^2 u \|^2$.
We use polar coordinates $(r, \phi)$. Then, $ (\nabla A) \cdot \nabla
= 2 r \partial_r$. 
Hence,
\begin{align*}
\|(\nabla A)\cdot \nabla u \|^2
=&
-4
\langle u ,
\partial_r( r^2 \partial_r u) \rangle
=
-4
\langle u ,
\left(
2 r \partial_r+ r^2 \partial_r^2
\right)
 u \rangle 
 \\
 =&
- 4
\langle u ,
r^2 \left(
r^{-1} \partial_r+ \partial_r^2
\right)
 u \rangle
 -4
 \langle u ,
r \partial_r
 u \rangle
 \\
 \leq &
  4
\langle r^2 u ,
- \left(
r^{-1} \partial_r+ \partial_r^2
+
\frac{1}{r^2} \partial^2_{ \phi }
\right)
 u \rangle
 -4
 \left\langle
|x| \frac{x}{|x|} 
  u ,
\nabla 
 u \right\rangle
 \\
 \leq &
 C
 (
 \| A u \|^2+
 \|-\Delta u \|^2
 +
 \| |x| u \|^2+ \|\nabla u \|^2
 )
 \,.
\end{align*}
Therefore,
$ \| u \|_{2,A} \leq C \| u \|_{{\Sigma,4}}$ holds. 
To show the converse,
first note that
$
 \| u \|_{{\Sigma,4}}^2
 \leq
 C( \|u \|^2+ \| A u \|^2+ \|-\Delta u \|^2
 +
 \|A^2 u \|^2+ \|\Delta^2 u \|^2)$. 
Since $-\Delta\leq -\Delta+ |x|^2$ and $|x|^2\leq -\Delta+ |x|^2$ holds as an operator inequality, we directly
obtain $ \| u \|_{{\Sigma,4}} \leq C \| u \|_{2,A}  $.

By $\|u v\|_{H^k} \leq \|u\|_\infty \|v\|_{H^k}+ \|u\|_{H^k} \|v\|_\infty$ ,
 $\| \cdot \|_{2,A}$ fulfills the generalized Leibniz rule
\begin{align*}
\| u v \|_{2,A}
\leq &
C\| u v \|_{{\Sigma,4}}
\leq
C (\|u \|_\infty \|v\|_{{\Sigma,4}}
+
 \|u\|_{{\Sigma,4}} \|v \|_\infty
 )
 \\
 \leq &
C(
\| u \|_{2,A} \|v \|_\infty+\|u\|_\infty \|  v \|_{2,A})
\;. 
\end{align*}

From \eqref{Nlsintegral}, we obtain 
\begin{align*}
\|\phi_t\|_{2,A}
\leq &
\|U_t \phi_0\|_{2,A}
+
|a|
\int_{0}^t ds
\|
U_{t-s} | \phi_s|^2 \phi_s
\|_{2,A}
\\
=&
\| \phi_0\|_{2,A}
+
|a|
\int_{0}^t ds
\|
 | \phi_s|^2 \phi_s
\|_{2,A}
\\
\leq &
\| \phi_0\|_{2,A}
+
C 
\int_{0}^t ds
\| \phi_s \|_\infty^2
\| \phi_s
\|_{2,A}
 \;.
\end{align*}
By a Gr\"onwall inequality, we obtain (b).
\item
We show that
$\phi_t \in H^2(\mathbb{R}^2, \mathbb{C})$ globally in time. 
Recall the existence of global in time solutions of 
\begin{align}
\label{NlsFree}
i \partial_t u_t= (-\Delta+a|u_t|^2) u_t \;.
\end{align}
in $H^2(\mathbb{R}^2, \mathbb{C})$, provided that $a>-a^*$ and $u_0 \in H^2(\mathbb{R}^2, \mathbb{C}), \|u_0\|=1$ holds.
Using the lens transform \cite{carles2, tao}, for $|t|< \pi/2$
\begin{align*}
\phi_t(x)=
\frac{1}{\cos(t)}
u_{\tan (t)}\left( \frac{x}{\cos(t)} \right)
e^{-i \frac{|x|^2}{2} \tan(t)} \;,
\end{align*} 
$\phi_t$ then solves $i \partial_t \phi_t=
(-\Delta+a|\phi_t|^2+ |x|^2) \phi_t$
with initial datum $\phi_0=u_0$.
 We therefore see that the existence of a global-in-time solution of 
\eqref{NlsFree} in $H^2(\mathbb{R}^ 2,\mathbb{C})$ implies existence of a solution
$\phi_t$ in $H^2(\mathbb{R}^ 2,\mathbb{C})$ locally in $t \in ]-\pi/2, \pi/2[$. By translation invariance of time, the solution $\phi_t$ then exists globally in $H^2(\mathbb{R}^ 2,\mathbb{C})$. 
By the embedding $
L^\infty(\mathbb{R}^ 2,\mathbb{C})
\subset
H^2(\mathbb{R}^ 2,\mathbb{C})
$, we obtain, together with (b), (c).

\end{enumerate}

\end{proof}

\subsection{Self-Adjointness} 
\label{selfadjoint}
\begin{lemma}
Let
\begin{align*}
H_{W_\beta, t}=
\sum_{k=1}^N (-\Delta_k)
+
\sum_{i<j=1}^N W_\beta (x_i-x_j)
+
\sum_{k=1}^N A_t(x_k)
\end{align*}
and assume (A1) and (A2).
Then, for all $t \in \mathbb{R}$,
\begin{enumerate}
\item $H_{W_\beta, t}$ is selfadjoint with domain
$ \mathcal{D}(H_{W_\beta, t})
=
\mathcal{D} \left(
\sum_{k=1}^ N(-\Delta_k+A_0(x_k))
\right)$.
\item
$(H_{W_\beta, t})^ 2$ is selfadjoint with domain
$\mathcal{D}((H_{W_\beta, t})^ 2)=
\mathcal{D}((H_{W_\beta, 0})^ 2)$. If, in addition, $W \in C^2(\mathbb{R}^2,\mathbb{R})$, then
$ \mathcal{D}((H_{W_\beta, t})^ 2)
=
\mathcal{D} \left(
(
\sum_{k=1}^ N(-\Delta_k+A_0(x_k))^2
\right)$ holds.

\end{enumerate} 
\end{lemma}
\begin{proof}
\begin{enumerate}

\item
First note that $ \mathcal{D}(H_{W_\beta, 0})
=
\mathcal{D} \left(
\sum_{k=1}^ N(-\Delta_k+A_0(x_k))
\right)$, since $W_\beta \in  L^{\infty}_c(\mathbb{R}^2, \mathbb{R})$.
We write
\begin{align*}
H_{W_\beta, t}
=
H_{W_\beta, 0}+ 
\sum_{k=1}^N
\int_0^t ds \dot{A}_s(x_k) \;.
\end{align*}
Abbreviate $
\mathcal{A}_t=
\sum_{k=1}^N
\int_0^t ds \dot{A}_s(x_k)$.
Since $ \| \mathcal{A}_t \Psi \| \leq  N \int_0^t ds \|\dot{A}_s\|_\infty
 \| \Psi \|$ holds for all $\Psi \in L^2(\mathbb{R}^2,\mathbb{C)}$, $\mathcal{A}_t$ is infinitesimal
 $H_{W_\beta, 0}$ bounded, which implies by Kato-Rellich that
$\mathcal{D}( H_{W_\beta, 0})=\mathcal{D}( H_{W_\beta, t})$.
\item
Note that $(H_{W_\beta, 0})^2$ is self-adjoint on $\mathcal{D}((H_{W_\beta, 0})^2)$.
Consider
\begin{align*}
(H_{W_\beta, t})^2
=&
(H_{W_\beta, 0})^2
+ H_{W_\beta, 0}\mathcal{A}_t
+
\mathcal{A}_t H_{W_\beta, 0} +
\mathcal{A}_t^2 \;.
\end{align*}
Under assumption (A2)
$
 H_{W_\beta, 0}\mathcal{A}_t
+
\mathcal{A}_t H_{W_\beta, 0} +
\mathcal{A}_t^2
$ is a  symmetric operator on $\mathcal{D}((H_{W_\beta, 0})^2)$.
We estimate, for $\Psi  \in \mathcal{D}((H_{W_\beta, 0})^2), \Psi \neq
 0
$,
\begin{align*}
&
\left\|
\left(
H_{W_\beta, 0}\mathcal{A}_t
+
\mathcal{A}_t H_{W_\beta, 0} +
\left(
\mathcal{A}_t
\right)^2
\right)
\Psi
\right\|
\\
\leq &
2 N
\int_0^t ds \|\dot{A}_s\|_\infty
\|
H_{W_\beta, 0}
\Psi\|
 + N^2
\left( \int_0^t ds \|\dot{A}_s\|_\infty\right)^2
\|
\Psi\|
\\
+&
 N
\int_0^t ds \|\Delta\dot{A}_s\|_\infty
\|
\Psi\|
+
2
\left\|
\sum_{k=1}^N
\int_0^t ds \nabla_k \dot{A}_s(x_k) \nabla_k \Psi 
\right\|
\end{align*}
Note that
\begin{align*}
2 N
&\int_0^t ds \|\dot{A}_s\|_\infty
\|
H_{W_\beta, 0}
\Psi\|
=
2 N
\int_0^t ds \|\dot{A}_s\|_\infty
\sqrt{\laa \Psi,
\left(H_{W_\beta, 0} \right)^2 \Psi \raa
}
\\
\leq &
\sqrt{
2 N^2
\left(
\int_0^t ds \|\dot{A}_s\|_\infty
\right)^2
\| \Psi \|^2+
\frac{1}{2}
\left\|\left(H_{W_\beta, 0} \right)^2 \Psi \right\|^2
}
\\
\leq &
\sqrt{2}N
\int_0^t ds \|\dot{A}_s\|_\infty
\| \Psi \|+
\frac{1}{\sqrt{2}}
\left\|\left(H_{W_\beta, 0} \right)^2 \Psi \right\|
\;.
\end{align*}
Furthermore, for $\epsilon>0$ 
\begin{align*}
&
2
\left\|
\sum_{k=1}^N
\int_0^t ds \nabla_k \dot{A}_s(x_k) \nabla_k \Psi 
\right\|
\leq
2
\sum_{k=1}^N
\int_0^t ds \|\nabla \dot{A}_s\|_{\infty}
\| \nabla_k \Psi \|
\\
\leq &
\frac{2N}{ \epsilon}
\left(\int_0^t ds \|\nabla \dot{A}_s\|_{\infty} \right)^2 \| \Psi \|
+
\frac{\epsilon}{2 \|\Psi\|}
\sum_{k=1}^N
\| \nabla_k \Psi \|^2
\\
\leq &
\frac{2N}{ \epsilon}
\left(\int_0^t ds \|\nabla \dot{A}_s\|_{\infty} \right)^2 \|\Psi\|
+
\frac{1}{2 \|\Psi\|}
\langle \Psi,
H_{W_\beta, 0}
\Psi \rangle
+
CN \| \Psi\|
\;.
\end{align*}
Since
$
\|\Psi\|^{-1}
\laa \Psi,
H_{W_\beta, 0}
\Psi \raa
\leq
\frac{1}{2}\| \Psi \|
+
\frac{1}{2 \|\Psi\|}
\|H_{W_\beta, 0}
\Psi \|^ 2
\leq
\| \Psi \|
+
\frac{1}{2 }
\|(H_{W_\beta, 0})^2
\Psi \|
$, we obtain
\begin{align*}
&\left\|
\left(
H_{W_\beta, 0}\mathcal{A}_t
+
\mathcal{A}_t H_{W_\beta, 0} +
\left(
\mathcal{A}_t
\right)^2
\right)
\Psi
\right\|
\leq 
\left(
\frac{1}{\sqrt{2}}
+
\frac{1}{4}
\right)
\| (H_{W_\beta, 0})^2 \Psi\|
\\
+&
 \left(
\sqrt{2}N \int_0^t ds \|\nabla \dot{A}_s\|_{\infty} 
+
\frac{2 N}{ \epsilon}
\left(\int_0^t ds \|\nabla \dot{A}_s\|_{\infty} \right)^2
+
CN
 \right)
\| \Psi\|
\\
+&
 N
\int_0^t ds \|\Delta\dot{A}_s\|_\infty
\|
\Psi\|
 \;.
\end{align*}
Thus, $H_{W_\beta, 0}\mathcal{A}_t
+
\mathcal{A}_t H_{W_\beta, 0} +
\left(
\mathcal{A}_t
\right)^2$ is relatively $(H_{W_\beta, 0})^2$ bounded with bound $
\frac{1}{\sqrt{ 2}}+ \frac{1}{4}<1$. By Kato-Rellich, $(H_{W_\beta, t})^2$ is self-adjoint with
domain $\mathcal{D}((H_{W_\beta, t})^2)=
\mathcal{D}((H_{W_\beta, 0})^2)$, for all $t \in \mathbb{R}$. By a similar estimate, we also obtain
$ \mathcal{D}((H_{W_\beta, 0})^ 2)
=
\mathcal{D} \left(
(
\sum_{k=1}^ N(-\Delta_k+A_0(x_k))^2
\right)$ if $W \in C^2(\mathbb{R}^2, \mathbb{C})$.

\end{enumerate}

\end{proof}

\section*{Acknowledgments}
We would like to thank Lea Bo{\ss}mann,  Nikolai Leopold and David Mitrouskas for many helpful discussions. We also would like to thank an anonymous  referee for various valuable comments on an earlier version of this paper.
M.J. gratefully acknowledges financial support by the German National Academic Foundation.


\begin{thebibliography}{}






\bibitem{anapolitanos}
I. Anapolitanos and M. Hott,
\emph{
A simple proof of convergence to the Hartree dynamics in Sobolev trace norms
}, J. Math. Phys. 57, 122108 (2016).

\bibitem{Bagnato}
V.  Bagnato and D.Kleppner
\emph{Bose-Einstein condensation in low-dimensional traps},
Phys. Rev. A 44, 7439 (1991).

\bibitem{benedikter}  N. Benedikter, G. De Oliveira and B. Schlein, \emph{Quantitative derivation of the Gross-Pitaevskii equation}, Comm. Pur. Appl. Math.  08 (2012).

\bibitem{brennecke}
C. Brennecke and B. Schlein
\emph{Gross-Pitaevskii Dynamics for Bose-Einstein Condensates}, 
arXiv:1702.05625 (2017).


\bibitem{SchleinNorm}
C. Boccato, S. Cenatiempo and B. Schlein,
\emph{Quantum many-body fluctuations around nonlinear Schr\"odinger dynamics},
arXiv:1509.03837 (2015).



\bibitem{carles2}
R. Carles
\emph{Nonlinear Schr\"odinger equation with time dependent potential},
Communications in mathematical sciences 9(4) (2009).

\bibitem{carles}
R. Carles and J. Drumond Silva,
\emph{
Large time behavior in nonlinear Schrodinger equation with time dependent potential}, Communications in Mathematical Sciences, International Press, 2015, 13 (2), pp.443-460.

\bibitem{cazenave}
T. Cazenave
\emph{Semilinear Schr\"odinger Equations },
Courant Lecture Notes, AMS (2003).




\bibitem{chen2d}
X. Chen and J. Holmer, \emph{ The Rigorous Derivation of the 2D Cubic Focusing NLS from Quantum Many-body Evolution}, Int Math Res Notices (2016).


\bibitem{cherny}
A. Yu. Cherny and A. A. Shanenko
\emph{Dilute Bose gas in two dimensions: density expansions and the Gross-Pitaevskii equation
}, PhysRevE.64.027105 (2001).

\bibitem{chong} 
J. J. W. Chong,
\emph{Dynamics of Large Boson Systems with Attractive Interaction and a Derivation of the Cubic Focusing NLS in $\mathbb{R}^3$}, arXiv:1608.01615 (2016).

\bibitem{erdos1} L. Erd\"os, B. Schlein and  H.-T. Yau, \emph{Derivation of the Gross-Pitaevskii Hierarchy for the Dynamics
of Bose-Einstein Condensate}, Comm.\ Pure Appl.\ Math.\ \textbf{59}
, no. 12, 1659--1741 (2006).

\bibitem{erdos2} L. Erd\"os, B. Schlein and  H.-T. Yau, \emph{ Derivation of the cubic non-linear Schr\"odinger equation from
quantum dynamics of many-body systems}, Invent. Math. 167
, 515--614 (2007).


\bibitem{erdos3} L. Erd\"os, B. Schlein and  H.-T. Yau, \emph{ Derivation of the Gross-Pitaevskii  equation for the dynamics of Bose-Einstein condensate}, Ann. of Math. (2) 172
, no. 1, 291--370 (2010).



\bibitem{erdos4} L. Erd\"os, B. Schlein and  H.-T. Yau, \emph{Rigorous derivation of the Gross-Pitaevskii  equation with a larger interaction potential}, J. Amer. Math. Soc. 22
, no. 4, 1099--1156 (2009).


\bibitem{experiment}
A. Görlitz, J. M. Vogels, A. E. Leanhardt, C. Raman, T. L. Gustavson, J. R. Abo-Shaeer, A. P. Chikkatur, S. Gupta, S. Inouye, T. Rosenband, and W. Ketterle
\emph{Realization of Bose-Einstein Condensates in Lower Dimensions},
Phys. Rev. Lett. 87, 130402 (2001).

\bibitem{guo}
 Y. Guo and R. Seiringer \emph{ On the mass concentration for Bose-Einstein condensates with
attractive interactions}, Lett. Math. Phys. 104,no. 2, 141–156  (2014).

\bibitem{jeblick}
M. Jeblick, N. Leopold and P. Pickl
\emph{
Derivation of the Time Dependent Gross-Pitaevskii Equation in Two Dimensions},
 arXiv:1608.05326 (2016).




\bibitem{ketterle}
W. Ketterle, \emph{Nobel lecture: When atoms behave as waves: Bose-Einstein condensation and the atom laser},
Rev. Mod. Phys. 74, no.~4, 1131--1151 (2002).


\bibitem{killip}
R. Killip, T. Tao and M. Visan
\emph{The cubic nonlinear Schrodinger equation in two dimensions with radial data},
J. Eur. Math. Soc. 11 (2009).

\bibitem{schlein2d}
K. Kirkpatrick, B. Schlein and Gigliola Staffilani, \emph{Derivation of the two-dimensional nonlinear Schr\"odinger equation from many body quantum dynamics}, American Journal of Mathematics 133, no. 1 (2011): 91-130.



\bibitem{knowles}A. Knowles and P.Pickl, \emph{Mean-Field Dynamics: Singular Potentials and Rate of Convergence},  Comm. Math. Phys.  298, 101-139 (2010).

\bibitem{leopold} N. Leopold and P. Pickl,
\emph{Derivation of the Maxwell-Schrödinger Equations from the Pauli-Fierz Hamiltonian
},  arXiv:1609.01545 (2016).

\bibitem{lewinMeanField}
M. Lewin,
\emph{Mean-Field limit of Bose systems: rigorous results},
Proceedings of the International Congress of Mathematical Physics (2015).


\bibitem{lewinneu}
 M. Lewin, Phan Th\'anh Nam and N. Rougerie,
  \emph{The mean-field approximation and the non-linear Schrödinger functional for trapped Bose gases}, Transactions of the American Mathematical Society 368, 6131-6157 (2016).

\bibitem{lewin}
M. Lewin, Phan Th\'anh Nam and N. Rougerie, 
\emph{A note on 2D focusing many-boson systems
},  Proc. Amer. Math. Soc.  (2016).

\bibitem{liebanalysis}
E. Lieb and M. Loss,
\emph{Analysis}, Graduate studies in mathematics, American Mathematical Society (2010).

\bibitem{lieb100bec}
E. Lieb and R.Seiringer,
\emph{Proof of Bose-Einstein condensation for dilute trapped gases.}, Phys Rev Lett. vol. 88, 170409 (2002).


\bibitem{lieb}
E. Lieb and R.Seiringer,
\emph{The Stability of Matter in Quantum Mechanics
}, Cambridge University Press, Cambridge (2010).


\bibitem{lssy}
E.H Lieb, R. Seiringer, J.P. Solovej and J. Yngvason, \emph{The
mathematics of the Bose gas and its condensation}, Oberwolfach
Seminars, {\bf 34} Birkhauser Verlag, Basel, (2005).

\bibitem{lsy}
E.H Lieb, R. Seiringer and J. Yngvason, \emph{ A Rigorous Derivation of the Gross-Pitaevskii Energy Functional for a Two-Dimensional Bose Gas
},
Commun. Math. Phys. 224, 17 (2001).

\bibitem{ls}
E.H. Lieb and J. Yngvason, 
\emph{ The Ground State Energy of a Dilute Two- dimensional Bose Gas}, J. Stat. Phys. 103, 509 (2001).






\bibitem{michelangeli}
A. Michelangeli, \emph{Equivalent definitions of asymptotic 100\% BEC}, Nuovo Cimento Sec. B.,  123, 181--192 (2008).

\bibitem{picklnorm}
D. Mitrouskas, S. Petrat  and P. Pickl,
\emph{Bogoliubov corrections and trace norm convergence for
the Hartree dynamics}, arXiv:1609.06264 (2016).


\bibitem{marcin1}
Phan Th\`anh Nam and M. Napi\'orkowski,
\emph{A note on the validity of Bogoliubov correction to mean-field dynamics},
arXiv:1604.05240 (2016).

\bibitem{marcin2}
Phan Th\`anh Nam and M. Napi\'orkowski,
\emph{Bogoliubov correction to the mean-field dynamics of interacting bosons},
arXiv:1509.04631 (2016).

\bibitem{marcin3}
Phan Th\`anh Nam and M. Napi\'orkowski,
\emph{Norm approximation for many-body quantum dynamics: focusing case in low dimensions
}, arXiv:1710.09684 (2017).


\bibitem{pickl2}
P.~Pickl, \emph{Derivation of the time dependent Gross-Pitaevskii equation without positivity condition on the interaction}, J. Stat. Phys. 140, 76--89 (2010).


\bibitem{picklgp3d}
P.~Pickl, \emph{Derivation of the time dependent Gross-Pitaevskii equation with external fields},  arXiv:1001.4894 Rev. Math. Phys., 27, 1550003 (2015).



\bibitem{pickl1}
P.~Pickl, \emph{A simple  derivation of mean field limits for quantum systems}, Lett. Math. Phys. 97, 151--164 (2011).




\bibitem{rodnianskischlein}
I.~Rodnianski and B.~Schlein, \emph{Quantum fluctuations and rate of
  convergence towards mean field dynamics}, Comm. Math. Phys.  291, no 1, 31--61 (2009).

\bibitem{tao}
T. Tao
\emph{A pseudoconformal compactification of the nonlinear Schr\"odinger equation and applications
}, New York Journal of Mathematics 15 (2006).




\bibitem{teschl}
G. Teschl \emph{Mathematical Methods in Quantum Mechanics
With Applications to Schr\"odinger Operators},
Graduate Studies in Mathematics, Volume 157, Amer. Math. Soc., Providence (2014).

\bibitem{weinstein}
M. Weinstein
\emph{ Nonlinear Schr\"odinger equations and sharp interpolation estimates},
 Comm.Math. Phys. 87, 567–576 (1983).


\end{thebibliography}
\end{document}